\def\full{1} %
\newcommand{\commentout}[1]{}
\ifnum\full=1
	\newcommand{\fullv}[1]{#1}
	\newcommand{\shortv}{\commentout}
\else
	\newcommand{\shortv}[1]{#1}
	\newcommand{\fullv}{\commentout}
\fi

\fullv{
\documentclass[12pt]{article}
\usepackage{setspace}
\setstretch{1.5}
\usepackage{chicagor}

\newcommand{\bbox}{\vrule height7pt width4pt depth1pt}
\def\qed{\bbox\vspace{0.1in}}

\setlength{\evensidemargin}{0.25in} 
\setlength{\oddsidemargin}{0.25in}
\setlength{\textwidth}{6.12in} 
\setlength{\textheight}{8.75in}
\setlength{\topmargin}{0.10in} 
\setlength{\headheight}{0in}
\setlength{\headsep}{0in} 
}

\shortv{
\documentclass{aamas2008}}

\newcommand{\authnote}[2]{{ \textbf{ [#1's Note:} {\em #2} \textbf{]} }}
\newcommand{\Rnote}[1]{{\authnote{Rafael}{#1}}}

\newcommand{\Play}{[n]}
\newtheorem{theorem}{Theorem}[section]
\newtheorem{lemma}[theorem]{Lemma}

\newtheorem{example}[theorem]{Example}

\newtheorem{proposition}[theorem]{Proposition}
\newtheorem{definition}[theorem]{Definition}
\fullv{
\newenvironment{proof}{\noindent\emph{Proof:}}{}}
\newcommand{\regret}{\mathit{regret}}
\newcommand{\Spure}{S}
\newcommand{\Smixed}{\Sigma}
\newcommand{\Scal}{\mathcal{S}}
\newcommand{\Tcal}{\mathcal{T}}
\newcommand{\B}{\mathcal{B}}
\newcommand{\D}{\mathcal{D}}
\newcommand{\Drm}{\mathcal{RM}}
\newcommand{\Dwd}{\mathcal{WD}}
\newcommand{\Dsd}{\mathcal{SD}}
\newcommand{\Dj}{\mathcal{J}}
\newcommand{\minregret}{\mathit{minregret}}
\newcommand{\Hcal}{\mathcal{H}}
\newcommand{\Exp}{\textbf{E}}
\newcommand{\ttm}{\mathit{TT}^{-}}
\newcommand{\gtrig}{\mathit{GT}}
\newcommand{\solve}{\mathit{solve}}
\newcommand{\MSW}{\mathit{MSW}}

\newlength{\saveparindent}
\newlength{\saveparskip}
\newcommand{\<}{\langle}
\renewcommand{\>}{\rangle}
\newcommand{\IR}{\mbox{$I\!\!R$}}

\newenvironment{oldthm}[1]{\par\noindent{\bf Theorem #1:} \em \noindent}{\par}
\newenvironment{oldlem}[1]{\par\noindent{\bf Lemma #1:} \em \noindent}{\par}
\newenvironment{oldcor}[1]{\par\noindent{\bf Corollary #1:} \em \noindent}{\par}
\newenvironment{oldpro}[1]{\par\noindent{\bf Proposition #1:} \em \noindent}{\par}
\newcommand{\othm}[1]{\begin{oldthm}{\ref{#1}}}
\newcommand{\eothm}{\end{oldthm} \medskip}
\newcommand{\olem}[1]{\begin{oldlem}{\ref{#1}}}
\newcommand{\eolem}{\end{oldlem} \medskip}
\newcommand{\ocor}[1]{\begin{oldcor}{\ref{#1}}}
\newcommand{\eocor}{\end{oldcor} \medskip}
\newcommand{\opro}[1]{\begin{oldpro}{\ref{#1}}}
\newcommand{\eopro}{\end{oldpro} \medskip}
\newcommand{\sigmaad}{s_{\mathit{ad}}}
\newcommand{\strat}{s}
\newcommand{\Strat}{S}
\newcommand{\strategy}{\mathit{strat}}

\newenvironment{newitemize}{%
\begin{list}{$\bullet$}{\labelwidth=18pt%
\labelsep=5pt \leftmargin=23pt \topsep=1pt%
\setlength{\listparindent}{\saveparindent}%
\setlength{\parsep}{\saveparskip}%
\setlength{\itemsep}{3pt} }}{\end{list}}

\begin{document}
\setcounter{page}{0}
\begin{titlepage}
\title{Iterated Regret Minimization: A More 
Realistic Solution Concept}
\author{
\fullv{
Joseph Y. Halpern\\
Cornell University\\
halpern@cs.cornell.edu 
\and Rafael Pass\\
Cornell University\\
rafael@cs.cornell.edu}}
\date{\today\thanks{First draft from October 2007. We thank 
Geir Asheim, Sergei Izmalkov,
Adam and Ehud Kalai, Silvio Micali, Henry Schneider, K\aa re Vernby,   and 
seminar participants at GAMES 2008 for helpful discussions.
Halpern is supported in part by NSF under grants ITR-0325453,
IIS-0534064, and IIS-0812045, and AFOSR. 
Pass is supported in part by an NSF CAREER Award CCF-0746990, AFOSR Award
FA9550-08-1-0197, and BSF Grant 2006317. 
}} 
\maketitle
\begin{abstract}
For some well-known games, such as the Traveler's Dilemma or the
Centipede Game, traditional game-theoretic solution concepts---and most
notably Nash equilibrium---predict outcomes that are not consistent
with empirical observations.
In this paper, we introduce a new solution concept, \emph{iterated
regret minimization},
which exhibits the same qualitative 
behavior as that observed in experiments in many games of interest, 
including Traveler's Dilemma, the Centipede Game, Nash bargaining, and
Bertrand competition. 
As the name suggests, iterated regret minimization involves the
iterated deletion of strategies that do not minimize regret.  
\end{abstract}
\thispagestyle{empty}
\end{titlepage}

\section{Introduction}
Perhaps the most common solution concept considered in game theory is
Nash equilibrium.  Various refinements of Nash equilibrium have been
considered, such as \emph{sequential equilibrium} \cite{KW82},
\emph{perfect equilibrium} \cite{Selten75}, and \emph{rationalizability}
\cite{Ber84,Pearce84}; see \cite{OR94} for an overview.  There are,
however, games where none of these concepts seems appropriate.  

Consider the well-known {\em Traveler's Dilemma} \cite{Basu94,Basu07}.
Suppose that two travelers have identical luggage, for which they both
paid the same price.  Their luggage is damaged (in an identical way)
by an airline.  The airline offers to recompense them for their luggage.
They may ask for any dollar amount between \$$2$ and \$100.
There is only
one catch.  If they ask for the same amount, then that is what they will
both receive.  However, if they ask for different amounts---say one asks
for \$$m$ and the other for \$$m'$, with $m < m'$---then whoever asks
for \$$m$ (the lower amount) will get \$$(m+p)$, 
while the other traveler will get \$$(m-p)$, where $p$ can be viewed 
as a reward for the person who asked for the lower amount, and a
penalty for the person who asked for the higher amount.

It seems at first blush that both travelers should ask for \$100,
the maximum amount, for then they will both get that.  The problem is
that one of them might then realize that he is actually better off
asking for \$99 if the other traveler asks for \$100, since he then
gets \$101.  In fact, \$99 \emph{weakly dominates} \$100, in that no
matter what Traveler 1 asks for, Traveler 2 is always at least as well
off asking for \$99 than \$100, and in one case (if Traveler 2 asks for
\$100) Traveler 1 is strictly better off asking for \$99.  Thus, it
seems we can eliminate 100 as an amount to ask for.  However, if we
eliminate 100, a similar argument shows that 98 weakly dominates 99!
And once we eliminate 
99, then 97 weakly dominates 98.  Continuing this
argument (technically, doing \emph{iterated deletion of weakly dominated
strategies}) both travelers end up asking for \$2!  In fact, it is easy 
to see that (2,2) is the only Nash equilibrium.
With any other pair of requests, at least one of the
travelers would want to change his request if he knew what the other
traveler was asking.  Since (2,2) the only Nash equilibrium, it is also
the only sequential and perfect equilibrium.  Moreover, it is the only
rationalizable strategy profile.  (It is not necessary to understand
these solution concepts in detail; the only point we are trying make
here is that all standard solution concepts lead to (2,2).)

This seems like a strange result.  It seems that no reasonable
person---even a game theorist!---would ever play 2.  Indeed, when 
the Traveler's Dilemma was empirically tested among game theorists (with
$p=2$) they typically  did not play anywhere close to 2.  Becker, Carter,
and Naeve \citeyear{BCN05} asked members of the Game
Theory Society (presumably, all experts in game theory) to 
submit a strategy for the game.  Fifty-one of them did so.  Of the 45
that submitted pure strategies, 33 submitted a strategy of 95 or higher,
and 38 submitted a strategy of 90 or higher; only 3 submitted the
``recommended'' strategy of 2.  The strategy that performed best (in
pairwise matchups against all submitted strategies) was 97, which had an
average payoff of \$85.09.  The worst average payoff went to those who
played 2; it was only \$3.92.

Another sequence of experiments by Capra et al.~\citeyear{CGGH99}
showed, among other things, that this result was quite sensitive to
the choice of $p$.  For low values of $p$, people tended to play high
values, and keep playing them when the game was repeated.  By way of
contrast, for high values of $p$, people started much lower, and
converged to playing 2 after a few rounds of repeated play.  
The standard solution concepts (Nash equilibrium, rationalizability,
etc.) are all insensitive to the choice of $p$; for example, (2,2) is
the only Nash equilibrium for all choices of $p > 1$.

In this paper, we introduce a new solution concept, \emph{iterated
regret minimization}, which
has the same qualitative behavior as that observed in the
experiments, 
not just in Traveler's Dilemma, but in many other games that have 
proved problematic for Nash equilibrium, including the Centipede Game,
Nash bargaining, and Bertrand competition.  In this
paper, we focus on iterated regret minimization in strategic games, and
comment on how it can be applied to Bayesian games. %

The rest of this paper is organized as follows. Section \ref{prel.sec}
contains preliminaries.  
Section~\ref{irm.sec} is the heart of the paper:
we 
first 
define iterated regret minimization in strategic games, 
provide an epistemic characterization of it,
and
then show how
iterated regret minimization works in numerous standard examples from
the game theory literature, including Traveler's Dilemma, Prisoner's
Dilemma, the Centipede Game, Bertrand Competition, and Nash Bargaining,
both with pure strategies and mixed strategies.   
The epistemic characterization, like those of many other solution
concepts, involves higher and higher levels of belief regarding other
players' rationality, it does \emph{not} involve common knowledge or
common belief.  Rather, higher levels of beliefs are accorded lower
levels of likelihood.
In Sections~\ref{sec:Bayesian} and \ref{mechdesign.sec}, we briefly consider
regret minimization in Bayesian games and in the context of mechanism
design. 
We discuss related work in Section \ref{sec:related}, and conclude in Section
\ref{discussion.sec}.  Proofs are relegated to the appendix.

\section{Preliminaries}
\label{prel.sec}
We refer to a collection of values, one
for each player, as a \emph{profile}.
If player $j$'s value is $x_j$, then the resulting profile is denoted
$(x_j)_{j \in \Play}$, or simply $(x_j)$ or $\vec{x}$, if 
the set of players is clear from the context.
Given a profile $\vec{x}$, let $\vec{x}_{-i}$ denote the  
collection consisting of all values $x_j$ for $j \ne i$. 
It is sometimes convenient to denote the profile $\vec{x}$ as
$(x_i, \vec{x}_{-i})$.  

A \emph{strategic game in normal form} is a ``single-shot'' game, where
each player $i$  
chooses an action from a space $A_i$ of actions.  For simplicity, we 
restrict our attention to \emph{finite} games---i.e., games where
the set $A_i$ is finite. Let $A = A_1 
\times \ldots \times A_n$ be the set of action profiles.  
A strategic game
is characterized by a tuple $(\Play,A,\vec{u})$, where $\Play$ is the
set of players, $A$ is the set of action profiles,
and $\vec{u}$ is the profile of utility functions, where $u_i(\vec{a})$ is
player $i$'s utility or payoff if the action profile $\vec{a}$ is played.  
A \emph{(mixed) strategy} for player $i$ is a probability distribution 
$\sigma_i \in \Delta(A_i)$ (where, as usual, we denote by $\Delta(X)$
the set of 
distributions on the set $X$).  
Let $\Smixed_i = \Delta(A_i)$ denote the mixed strategies for player $i$
in game 
$G$, and let 
$\Smixed = \Smixed_1 \times \cdots \times \Smixed_n$ denote the
set of mixed strategy profiles.
Note that, in strategy profiles in $\Smixed$, players are randomizing
independently.  
A \emph{pure strategy} for player $i$ is a strategy for $i$ that assigns
probability 1 to a single action. To simplify notation, we let an action
$a_i \in A_i$ also denote the pure strategy $\sigma_i \in \Delta(A_i)$
which puts weight only on 
$a_i$. 
If $\sigma$ is a strategy for player $i$ then $\sigma(a)$ denotes the
probability given to action $a$ by strategy $\sigma$. Given a strategy profile
$\vec{\sigma}$, player $i$'s expected utility if $\vec{\sigma}$ is
played, denoted $U_i(\vec{\sigma})$, 
is $\Exp_{\Pr}[u_i^{\vec{\sigma}}]$, 
where the expectation is taken with
respect to the probability $\Pr$ induced by $\vec{\sigma}$
(where the players are assumed to choose their actions independently).

\commentout{
A \emph{Nash equilibrium} for a 
strategic 
game $(\Play,A,\vec{u})$ is a 
(possibly mixed)
strategy profile $\vec{\sigma}$ such that for all $i \in \Play$
and $a_i \in A_i$, $U_i(\sigma_i, \vec{\sigma}_{-i}) \geq
U_i(a_i,\vec{\sigma}_{-i})$.  
}

\section{Iterated Regret Minimization in Strategic Games}
\label{irm.sec}
We start by providing an informal discussion of iterated regret
minimization in strategic games, 
and applying it to the Traveler's Dilemma; we then give a more
formal treatment.

Nash equilibrium implicitly assumes that the players know what strategy
the other players are using.  (See \cite{AB95} for a discussion of the
knowledge required for Nash equilibrium.)  Such knowledge seems
unreasonable, especially in one-shot games.  Regret minimization is one
way of trying to capture the intuition that a player wants to do well no
matter what the other players do.

The idea of minimizing regret was introduced (independently) in decision
theory by Savage \citeyear{Savage51} and Niehans \citeyear{Niehans}.  To
explain how we use it in a game-theoretic context, we first
review how it works in a single-agent decision problem.  Suppose that an
agent chooses an act from a set $A$ of acts.  The
agent is uncertain as to the true state of the world; there is a set
$S$ of possible states.  Associated with each state $s \in S$ and act $a
\in A$ 
is the utility $u(a,s)$ of performing act $a$ if
$s$ is the true 
state of the world.  For simplicity, we take $S$ and $A$ to be finite here.
The idea behind the \emph{minimax regret}
rule is to hedge the agent's bets, by doing reasonably well no matter
what the actual state is.   
For each state $s$, let $u^*(s)$ be the best outcome in state $s$; 
that is, $u^*(s) = \max_{a\in A}u(a,s)$.  
The {\em regret\/} of $a$ in 
state  $s$, denoted $\regret_u(a,s)$, is
$u^*(s) - u(a,s)$; that 
is, the regret of $a$ in $s$ is 
the difference between 
the utility of the best possible outcome in $s$ and the utility of
performing act $a$ in $s$.
Let $\regret_u(a) = \max_{s \in S} \regret_u(a,s)$.
For example, if $\regret_u(a) = 2$, then in each state $s$, 
the utility of performing $a$ in $s$ is guaranteed to be
within $2$ of the utility 
of any act the agent could choose, even if
she knew that the actual state was $s$.
The minimax-regret decision rule
orders acts by their regret; the ``best'' act is the
one that minimizes regret.  Intuitively, this rule is trying
to minimize the regret that an agent would feel if she discovered what the
situation actually was:~the ``I wish I had chosen $a'$ instead of
$a$'' feeling.  

Despite having been used in decision making for over 50 years, up until
recently, there seems to have been no attempt to apply regret
minimization in the context of game theory.  
We discuss other recent work on applying regret minimization to game 
theory in Section~\ref{sec:related}; here, we describe our own approach.
For ease of
exposition, we start by explaining it in the context of the Traveler's
Dilemma. 
We take the acts for one player to be that player's pure strategy
choices %
and take the states to be the other player's pure
strategy choices.  
Each act-state pair is then just a strategy profile;
the utility of the act-state pair for player $i$ is just the payoff to
player $i$ of the strategy profile.  
Intuitively, each agent is uncertain about what the other agent will do,
and is trying to choose an act that will minimize his regret, given that
uncertainty.

It is easy to see that, if the penalty/reward $p \le 49$, then the acts that
minimize regret are the ones in the interval $[100-2p,100]$;
the regret for all these acts is $2p-1$.  For if the other player asks for
\$$m \le 100-2p$, then the best response is to ask for \$$m-1$, which
results in a payoff of \$$m-1 + p$, while asking for $m' \in
[100-2p,100]$ %
results in a payoff of $m-p$, so the regret is $2p-1$.  
The regret may be less (but will not
be more) if the other player asks for an amount $m \in [100-2p, 100]$.
It is also easy to see that every other strategy has higher regret.
For example, if the other player plays 100, then the best response is
99; the regret of someone who chooses $m < 99$ is $99-m$, so if $m <
100-2p$, then the regret is greater than $2p-1$.  On the other hand, if
$p \ge 50$, then the unique act that minimizes regret is asking for \$2.   

Suppose that $p \le 50$.  
Applying regret minimization once suggests that we consider a strategy in
the interval  $[100-2p,100]$.  But we can iterate this process.  If
we assume that both players use a strategy in this interval, then
the strategy that minimizes regret is that of asking for \$$(100-2p+1)$.  A
straightforward check shows that this has regret $2p-2$; all
other strategies have regret $2p-1$.  In the special case that $p=2$,
this approach singles out the strategy of asking for \$97, which was
found to be the best strategy by Becker, Carter,
and Naeve \citeyear{BCN05}.  

As $p$ increases,  the act that survives this iterated deletion process
goes down, reaching $2$ if $p \ge 50$.  This matches, at a
qualitative level, the findings of Capra et
al.~\citeyear{CGGH99}.\footnote{
Capra et al.~actually considered a slightly different game where the
minimum bid was $p$ (rather than 2).  
If we instead consider this game, we get an even closer qualitative
match to their experimental observations.} 
\subsection{Deletion Operators and Iterated Regret Minimization}
Iterated regret minimization proceeds much like other notions of
iterated deletion.  To put it in context, we first abstract the notion
of iterated deletion.

Let $G=(\Play,A,\vec{u})$ be a strategic game.
We define iterated regret minimization in a way that 
makes it clear how it relates to other solution
concepts based on iterated deletion.
A \emph{deletion operator} $\D$ maps sets $\Scal = \Scal_1 \times \cdots
\times \Scal_n$ of strategy profiles in $G$ 
to sets of strategy profiles such that $\D(\Scal) \subseteq \Scal$.
Moreover, $\D(\Scal) = 
\D_1(\Scal)\times \cdots \times \D_n(\Scal)$, where $\D_i$ maps sets of
strategy profiles to strategies for player $i$.  
Intuitively, $\D_i(\Scal)$ is  the set of 
strategies for player $i$ that survive deletion, given that we start
with $\Scal$. Note that the set of strategies that survive deletion
may depend on the set that we start with.  Iterated deletion then amounts
applying the $\D$ operator repeatedly, starting with an appropriate
initial set $\Scal_0$ of strategies, where $\Scal_0$ is typically either
the set of pure strategy profiles (i.e., action profiles) in $G$ or the
set of mixed strategy profiles in $G$.

\begin{definition} \emph{ Given a deletion operator $\D$ and an initial set
$\Scal_0$ of strategies, the set of strategy profiles that
\emph{survive 
iterated deletion with respect to $\D$ and $\Scal_0$} is
$$\D^\infty(\Scal_0) = \cap_{k > 0}
\D^k(\Scal_0)$$ (where $\D^1(\Scal) = \D(\Scal)$ and
$\D^{k+1}(\Scal) = 
\D(\D^k(\Scal)$).  Similarly, the set of  
\emph{strategy profiles for player $i$ that survive
iterated deletion with respect to $\D$ and $\Scal_0$} is 
$\D_i^\infty(\Scal_0)  = \cap_{k > 0} \D_i^k(\Scal_0)$,
where $\D_i^1 = \D_i$ and $\D_i^{k+i} = \D_i \circ \D^k$.
}
\end{definition}

We can now define the deletion operator $\Drm$ appropriate for regret
minimization 
in strategic games  (we deal with Bayesian games in
Section~\ref{sec:Bayesian}).    
Intuitively, $\Drm_i(\Scal)$ consists of all the
strategies in $\Scal_i$ that minimize regret, given that the other
players are using a strategy in $\Scal_{-i}$.  
In more detail, we proceed as follows.  Suppose that 
$G$ is a strategic game $(\Play,A,\vec{u})$ and that
$\Scal 
\subseteq A$, the set of pure strategy profiles (i.e., actions).
For $\vec{a}_{-i} \in \Scal_{-i}$, let
$u^{\Scal_i}_i(\vec{a}_{-i}) = \max_{a_i\in \Scal_i} u_i(a_i,\vec{a}_{-i})$.
Thus, 
$u^{\Scal_i}_i(\vec{a}_{-i})$ 
is the best
outcome for $i$ given that the remaining players play $\vec{a}_{-i}$ and
that $i$ can select actions only in $\Scal_i$.
For $a_i \in \Scal_i$ and $\vec{a}_{-i} \in \Scal_{-i}$, let
the {\em regret\/} of $a_i$ for player $i$ given $\vec{a}_{-i}$ relative
to $\Scal_i$, denoted
$\regret^{\Scal_i}_{i}(a_i \mid \vec{a}_{-i})$, be 
$u_i^{\Scal_{i}}(\vec{a}_{-i}) - u_i(a_i,\vec{a}_{-i})$.
Let $\regret_{i}^{\Scal}(a_i) = \max_{\vec{a}_{-i} \in \Scal_{-i}}
\regret^{\Scal_{i}}(a_i \mid \vec{a}_{-i})$ denote the 
maximum regret of
$a_i$ for 
player $i$
(given that the other players' actions are chosen from $\Scal_{-i}$).
Let $\minregret^{\Scal}_i = \min_{a_{i} \in \Scal_i}
\regret^{\Scal_{-i}}_{i}(a_i)$ be the minimum regret for player
$i$ relative to $\Scal$. Finally, let   
$$\Drm_i(\Scal) = \{a_i \in \Scal_i : \regret^{\Scal}_{i}(a_i)
= \minregret^{\Scal}_i \}.$$ 
Thus, $\Drm_i(\Scal)$ consists of 
the set of actions that achieve the minimal regret with respect
to $\Scal$.  Clearly $\Drm_i(\Scal) \subseteq \Scal$.
Let $\Drm(\Scal) = \Drm_1(\Scal) \times \cdots \times \Drm_n(\Scal)$.

If $\Scal$ consists of mixed strategies, then the construction of
$\Drm(\Scal)$ is the same, except that the expected utility operator $U_i$
is used rather than $u_i$ in the definition of $\regret_i$.
We also need to argue that there is a strategy $s_i$ for player $i$
that maximizes $\regret_i$ and one that minimizes $\minregret_i$.  This
follows from the compactness of the sets of which the max and min are
taken, and the continuity of the functions being maximized and minimized.

\begin{definition} \emph{
Let $G= (\Play,A,\vec{u})$ be a strategic game.
$\Drm_i^\infty(A)$ is the set of \emph{(pure) strategies for player
$i$ that survive
iterated regret minimization with respect to pure strategies in $G$}.
Similarly,  
$\Drm_i^\infty(\Smixed(A))$ is the set of \emph{(mixed) strategies for
player $i$ that survive
iterated regret minimization with respect to mixed strategies in $G$}.
}
\end{definition}

\shortv{The following theorem}
\fullv{The following theorem, whose proof is in the appendix,}
shows that iterated regret minimization is a
reasonable concept in that, for all games $G$, 
$\Drm^{\infty}(A)$ and $\Drm^{\infty}(\Smixed(A))$ are nonempty fixed
points of 
the deletion process, that is, $\Drm(\Drm^{\infty}(A)) =
\Drm^{\infty}(A)$ and $\Drm(\Drm^{\infty}(\Smixed(A))) =
\Drm^{\infty}(\Smixed(A))$; the deletion process converges at
$\Drm^{\infty}$.
(Our proof actually shows that for any
nonempty closed set $\Scal$ of strategies, the set
$\Drm^{\infty}(\Scal)$ is nonempty and is a fixed point of the deletion
process.)  

\begin{theorem}
\label{nonempypure.prop}
Let $G=(\Play,A,\vec{u})$ be a strategic game. 
If $\Scal$ is a closed, nonempty set of strategies of the form
$\Scal_1 \times \ldots \times \Scal_n$, then 
$\Drm^\infty(\Scal)$ is nonempty, $\Drm^\infty(\Scal) = \Drm^\infty_1(\Scal)
\times \ldots \times \Drm^\infty_n(\Scal)$, and 
$\Drm(\Drm^\infty(\Scal))
= \Drm^\infty(\Scal)$.
\end{theorem}

Unlike standard solution concepts that involve equilibrium and,
implicitly, knowledge of the other agents' strategies, in a strategy
profile that survives iterated regret minimization, a player is not
making a best response to the strategies used by the other players
since, intuitively, he does not know what these strategies are.  
As a result, a player chooses a strategy that ensures that he does
reasonably well compared to the best he could have done, no matter what
the other players do.
We shall see the impact of this in the examples of
Section~\ref{sec:examples}.  

\shortv{
\begin{proof}
The argument is trivial in the case that $\Scal$ consists of only pure
strategies.  Note $\Drm^{k+1}(\Scal) \subseteq \Drm^k(\Scal)$.  Since
$\Scal$ is finite, there must be some 
$k$ for which $\Drm^{k+1}(\Scal) = \Drm^k(\Scal)$.  In this case, we
also have $\Drm^{\infty}(\Scal) = \Drm^k(\Scal)$, and the result
easily follows.

For the general case, the first step is to show that for all closed
nonempty sets $\Scal$, $\Drm(\Scal)$ is a closed, nonempty set.  The
argument then uses the fact that $\Scal$ can be viewed as a compact
subspace of $[0,1]^{|A|}$.  Details are available in the full paper
(an appropriately blinded version of which
can be obtained by contacting the program chair).
\qed
\end{proof}
}

\commentout{
\end{proof}
We also show that iterated regret minimization w.r.t mixed strategies results
in a non-empty set.
\begin{proposition}
Let $G=(\Play,A,\vec{u})$ be a strategic game. 
Then for every $i \in \Play$, 
$\Smixed_i^{\infty}$ is non-empty.
\end{proposition}
\begin{proof} (sketch) 
As in proof of proposition \ref{nonempypure.prop} it directly follows that 
$\Smixed_i^j$ is non-empty for every $j \in N$.
Since $regret$ is .., is also follows inductively that $\Smixed_i^j$ is
closed\Rnote{this part needs to be prove, i.e the closeness};
additionally it is clearly bounded. We conclude that  
${\Smixed_i^j}_{j \in N}$ is a nested sequence of closed and bounded sets, and thus by Cantor's infinite intersection theorem, $\Smixed_i^{\infty}$ is non-empty.
\qed \end{proof}
}

\subsection{Comparison to Other Solution Concepts
Involving Iterated Deletion
}
Iterated deletion has been applied in other solution concepts.
We mention three here.  
Given a set $\Scal$ of strategies, 
a strategy $\sigma \in \Scal_i$ is {\em weakly dominated\/} by
$\tau \in \Scal_i$ with respect to $\Scal$ if, for some strategy
$\vec{\sigma}_{-i} \in  \Scal_{-i}$, we have
$U_i(\sigma,\vec{\sigma}_{-i}) < U_i(\tau, \vec{\sigma}_{-i})$ and,
for all strategies $\vec{\sigma}'_{-i} \in \Scal_{-i}$, we have
$U_i(\sigma, \vec{\sigma}'_{-i}) \le U_i(\tau, \vec{\sigma}'_{-i})$.
Similarly, $\sigma$ is \emph{strongly dominated\/} by $\tau$ with
respect to $\Scal$ if 
$U_i(\sigma, \vec{\sigma}'_{-i}) < U_i(\tau, \vec{\sigma}'_{-i})$ 
for all strategies $\vec{\sigma}'_{-i} \in \Scal_{-i}$.
Thus,
if $\sigma$ is weakly dominated by $\tau$ with respect to
$\Scal_{-i}$,
then $i$ always does at least as well with $\tau$ as
with $\sigma$, and sometimes does better (given that we restrict
to strategies in $\Scal_{-i}$); if $\sigma$ is strongly dominated by
$\tau$, then player $i$ always does better with $\tau$ as with $\sigma$.
Let $\Dwd_i(\Scal)$ (resp., $\Dsd_i(\Scal)$) consist of all strategies
$\sigma_i \in \Scal_i$ that are not weakly (resp., strongly) dominated
by some strategy in $\Scal_i$ with respect to $\Scal$.   We can then
define the \emph{pure strategies that survive iterated weak (resp.,
strong) deletion with respect to pure strategies} as $\Dwd(A)$
(resp., $\Dsd(A)$).  And again, we can start with $\Smixed$ to get
corresponding notions for mixed strategies.

As is well known \cite{OR94}, the \emph{rationalizable} strategies
can also be considered as the outcome of an iterated deletion process.
Intuitively, a pure strategy for player $i$ is   
rationalizable if it is a best response to some beliefs that player $i$    
may have about the pure strategies that other players are following.
Given a set $\Scal$ of pure strategy profiles,
$\sigma \in \Scal_i$ is \emph{justifiable} if 
there
is some distribution $\mu$ on the strategies
in $\Scal_{-i}$ such that $\sigma$ is a best response to the resulting
mixed strategy.  Intuitively, $\mu$ describes player $i$'s beliefs about
the likelihood that other players are following various strategies;
thus, a strategy $\sigma$ for $i$ is justifiable
if there are beliefs that $i$ could have to which $\sigma$ is a best
response.  Let $\Dj_i(\Scal)$ consist of all strategies for player $i$
that are justifiable with respect to $\Scal$.
A pure strategy $\sigma$ for player $i$ is \emph{rationalizable} if $\sigma
\in \Dj^\infty_i(A)$.%
\footnote{The notion of rationalizability is typically applied to pure
strategies, although the definitions can be easily extended to deal with
mixed strategies.}

\subsection{An Epistemic Characterization of Iterated Regret
Minimization}\label{sec:characterization}

It is well known that rationalizability can be characterized in terms
of common knowledge of rationality \cite{TW88}, where a player is
rational if he has some beliefs according to which what he does is a
best response.   Thus, 
it is common knowledge among the players that 
the rationalizable strategies are a best
response to some
beliefs whose support is the set of strategies that remain after iterated
deletion.  

At first blush, it may seem that other notions of iterated deletion can
also be characterized in terms of common
knowledge.  Intuitively, if it is common knowledge that
all players are using the same deletion process, starting 
with
a commonly
known set of initial strategy profiles.  Since they are all intelligent,
they all delete once.  Realizing that they have all deleted once, they
delete  again with respect to the smaller set of strategy profiles;
realizing this, they delete again; and so on. 

This intuition is essentially true in the case of 
iterated deletion of weakly dominated strategies, although making it
precise is somewhat more subtle.  The justification for deleting a
weakly dominated strategy is the existence of other strategies.  But
this justification may disappear in later deletions. As Mas-Colell,
Whinston, and Green \citeyear[p. 240]{MWG95} put in their textbook when
discussing iterated deletion of weakly dominated strategies:
\begin{quote}
[T]he argument for deletion of a weakly dominated strategy for player
$i$ is that he contemplates the possibility that every strategy
combination of his rivals occurs with positive probability.  However,
this hypothesis clashes with the logic of iterated deletion, which
assumes, precisely, that eliminated strategies are not expected to
occur.
\end{quote}

Brandenburger, Friedenburg, and Kiesler \citeyear{BFK04} resolve this
paradox in the context of iterated deletion of weakly dominated strategies
by assuming that strategies were not really eliminated.  Rather, they
assumed that strategies that are weakly dominated occur with
infinitesimal (but nonzero) probability.   This is formally modeled in a
framework where uncertainty is captured using a lexicographic
probability system (LPS) \cite{BBD1}, whose support consists of all types.
(Recall that an LPS is a sequence $(\mu_0, \mu_1, \ldots)$ of
probability measures, in this case on type profiles, where $\mu_1$ 
represents events that have infinitesimal probability
relative to $\mu_0$, $\mu_1$ represents events that have infinitesimal
probability relative to $\mu_1$, and so on.  Thus, a probability
of $(1/2, 1/3, 1/4)$ can be identified with a nonstandard probability of
$1/2 + \epsilon/3 + \epsilon^2/4$, where $\epsilon$ is an infinitesimal.) 
In this framework, they show that iterated deletion of weakly dominated
strategies corresponds to what they call \emph{common assumption} of
rationality, where ``common assumption'' is a variant of ``common
knowledge'', and ``rationality'' means ``does not play a weakly
dominated strategy''.

With iterated regret minimization, there is a conceptual problem
somewhat similar to that of 
iterated deletion of weakly dominated strategies.
For simplicity, suppose that there are only two players, 1 and 2, and
that the strategy profiles $\Scal_1$ for player 1 and $\Scal_2$ for
player 2 survive iterated deletion.  Although the strategies in $\Scal_1$
minimize regret with respect to $\Scal_2$ \emph{among the
strategies in $\Scal_1$}, it may well be the case that some of the
deleted strategies might have even lower regret with respect to the
strategies in $\Scal_2$.  Indeed, this is the case with the Traveler's
Dilemma.  As we observed above, the strategy profile
$(97,97)$ is the only one that survives iterated regret minimization
when $p=2$.  However, if agent 1 knows that player 2 is playing 97, then
he should play 96, not 97!  That is, among all strategies, 97 is
certainly not the strategy minimizes regret with respect to $\{97\}$.  

The approach taken by Brandenburger, Friedenberg, and
Keisler \citeyear{BFK04} does not seem to help in resolving this
problem.  Assigning deleted 
strategies infinitesimal probability will not make 97 a best response to
a set of strategies where 97 is given very high probability.
We deal with this problem by essentially reversing the approach taken by
Brandenburger, Friedenberg, and Keisler.
Rather than assuming common knowledge of rationality, we 
assign successively lower probability to higher orders of
rationality. 
Roughly speaking, the idea is that now, with overwhelming probability,
no assumptions are made about the other players; with probability
$\epsilon$, they are assumed to be rational, with probability
$\epsilon^2$, the other players are assumed to be rational and to
believe that they are playing rational players, and so on.
(Of course, ``rationality'' is interpreted here as minimizing expected
regret.)  
This approach is consistent with the spirit of Camerer, Ho, and Chong's
\citeyear{CHC04} \emph{cognitive hierarchy}
model, where the fraction of people with $k$th-order beliefs declines as
a function of $k$, although not as quickly as this informal discussion
suggests.
 
Since regret minimization is non-probabilistic, 
the formal model 
of a lexicographic belief
is a
sequence $(\Scal^0,\Scal^1,\ldots)$ of sets of strategy profiles.
The strategy profiles in $\Scal^0$ represent the players' primary
beliefs, the strategy profiles in $\Scal^1$ are the players' secondary
beliefs, and so on.  (We can think of $\Scal^k$ as the support of the
measure $\mu_k$ in an LPS.)%
\footnote{Like LPS's, this model implicitly assumes, among other things,
that players $i$ and $j$ have the same beliefs about players $j' \notin
\{i,j\}$.  This assumption is acceptable, given that we assume that 
(it is commonly known that) all players start the iteration process 
by considering all strategies.  To get an epistemic characterization of
a more general setting, where players' initial beliefs about other
players strategies are not commonly known, we need a slightly more
general model of beliefs, where each player has his or her own
lexicographic sequence; see Section~\ref{sec:prior}.}
We call $\Scal^i$ the \emph{level-$i$ belief} of the lexicographic
belief $(\Scal^0,\Scal^1,\ldots)$.

Given such lexicographic beliefs, what strategy should a rational player $i$
choose?  Clearly the most important thing is to minimize regret with
respect to his primary beliefs, $\Scal^0_{-i}$.  But among strategies
that minimize regret with respect to $\Scal^0_{-i}$, the best are those
strategies that also minimize regret with respect to $\Scal^1_{-i}$;
similarly, among strategies that minimize regret with respect to each of 
$\Scal^1_{-i}, \ldots, \Scal^{k-1}_{-i}$, the best are those that also
minimize regret with respect to $\Scal^k_{-i}$.
Formally, a strategy $\sigma$ for player $i$ is \emph{rational with respect to
a lexicographic sequence $(\Scal^0,\Scal^1, \ldots )$} if there exists a
sequence $(\Tcal^0, \Tcal^1, \ldots)$ of strategy profiles such that 
$\Tcal^0_i$ consists of all strategies $\tau$ such that $\tau_i$
minimizes regret with respect to $\Scal^0_{-i}$ for all players $i$;
and
$\Tcal^k$ for $k > 0$ is defined inductively to consist of all
strategies $\tau \in \Tcal^{k-1}$ such that $\tau_i$ has the least
regret with respect to $\Scal^k_{-i}$ among all strategies in
$\Tcal^{k-1}_i$;
and $\sigma \in \cap_{k=0}^\infty \Tcal_i^k$.%
\footnote{A straightforward argument by induction shows that $\Tcal^k$ is
nonempty and compact, so that there will be a strategy $\tau_i$ that has
the least regret with respect to $\Scal^k_{-i}$ among all strategies in
$\Tcal^{k-1}_i$.}
Of course, this definition makes perfect sense if the
lexicographic sequence is finite and has the form $(\Scal^0, \ldots,
\Scal^k)$; in that case we consider $(\Tcal^0,\ldots, \Tcal^k)$.
Such a sequence $(\Scal^0, \ldots, \Scal^k)$ is called a
\emph{$(k+1)$st-order lexicographic belief}.  
It easily follows 
that $\emptyset \ne \cdots \subseteq \Tcal^k \subseteq \Tcal^{k-1}
\subseteq \cdots \subseteq \Tcal^0$, so that a strategy that is rational
with respect to $(\Scal^0, \Scal^1, \ldots)$ is also rational with
respect to each of the finite prefixes $(\Scal^0, \Scal^1, \ldots)$.  

Up to now, we have not imposed
any constraints on \emph{justifiability}
of beliefs. 
We provide a recursive definition of justifiability.
A ($k$th-order) lexicographic belief $(\Scal_j)_{j \in I}$ is 
\emph{justifiable} if, for each $j \in I$, the level-$j$ belief $\Scal_j$ is
\emph{level-$j$ justifiable}, where level-$j$ justifiability is
defined as follows. 
%
%
\begin{itemize}
\item To capture the intuition that players' primary beliefs are such that they make no assumptions about the other players,
we say that 
a belief $\Scal^0_i$ is \emph{level-0 justifiable} if it is the full set 
of strategies $\Scal_i$ available to player $i$.\footnote{As we discuss in Section
\ref{sec:prior}, we can also consider a more general model where players
have prior beliefs; in such a setting, $\Scal^0_i$ need not be the full
set of strategies.} 
\item To capture the intuition that players' level-$k$ belief is that 
the other players are $(k-1)$st-order rational, 
we say that a belief, $\Scal^k_i$, is \emph{level-$k$ justifiable} if there
exists some justifiable $k$th-order belief $(\Scal'^{0}_{-i},
\Scal'^{1}_{-i}, \ldots, \Scal'^{k-1}_{-i})$ such that  
$\Scal^k_i$ is the set of rational strategies for player $i$ with
respect to $(\Scal'^{0}_{-i}, \Scal'^{1}_{-i},
\ldots, \Scal'^{k-1}_{-i})$. 
\end{itemize}
This notion of justifiability captures the intuition 
that, with probability $\epsilon^k$,
each player $j_k$ believes that each other player $j_{k-1}$ is 
rational with respect to a $k$th-order belief and 
believes that, with
probability $\epsilon^{k-1}$, each other player $j_{k-2}$ is rational
with respect to a $(k-1)$st-order belief
and believes that, 
with probability $\epsilon^{k-2}$, \ldots, and with probability
$\epsilon$ believes that each other player $j_1$ is rational with
respect to a first-order belief and believes that 
each other player $j_0$ is playing an arbitrary strategy in $\Scal^0$
(As usual, ``rationality'' here means ``minimizes regret with respect
to his beliefs''.)   

Given these definition, we have the following theorem.
\begin{theorem}
Let $G = ([n],A, \vec{u})$ be a strategic game and let $\Scal$ be the full set of pure or mixed strategies.
Then
for each $k \in N$ there exists a unique level-$k$ justifiable belief
$\tilde{\Scal}^k = \Drm^{k-1}(\Scal)$. Furthermore,
$\Drm^{\infty}(\Scal)$  
is the set of rational strategies with respect to the belief
$(\tilde{\Scal}^0, \tilde{\Scal}^1, \ldots)$ and $\Drm^k(\Scal)$ is the
set of rational strategies with respect to the 
belief $(\tilde{\Scal}^0, \tilde{\Scal}^1, \ldots, \tilde{\Scal}^k)$
\end{theorem}
\begin{proof}
By definition there is a unique level-$0$ justifiable belief
$\tilde{\Scal^0}=\Scal$. It inductively follows that there exists a
unique level-$k$ justifiable belief $\tilde{\Scal^k}=\Drm(\Scal^{k-1}) =
\Drm^{k-1}(\Scal)$. The theorem then follows from the definition of
rationality with respect to a lexicographic belief. 
\qed
\end{proof}

\noindent Note that the sets 
$(\tilde{\Scal}^1, \tilde{\Scal}^2, \ldots)$ are just the sets 
$(\Tcal^0, \Tcal^1, \ldots)$ given in the definition of rationality.

In Appendix \ref{sec:kripke}, we provide an alternative
characterization of iterated regret minimization in terms of Kripke
structures.  
%

%
%
%
%
\commentout{
In the context of iterated regret 
minimization, $\Scal^k$ has the form $\Scal^k_1 \times \cdots \times
\Scal^k_n$, and $\Scal^0 \supseteq \Scal^1 \supseteq \Scal^2 \supseteq
\cdots$.  The profiles in $\Scal^0_{i}$ 
are the ones that player $i$ is (commonly believed to be) choosing
among, and the strategies in $\Scal^{k+1}_i$ are the ones that player  
$i$ would choose if he is rational (i.e., minimizes regret) and believes
that the other players choose a strategy profile in $\Scal^k_{-i}$.  Thus, 
for example, $\Scal^2_i$ consists of the strategies that player $i$ would
choose if he is rational, and believes that each other player $j$ is
rational and that each other player $j$ believes that every other player
$j'$ is choosing from $\Scal^0_{j'}$.  

However, we do \emph{not} assume
that, for example, player $i$ actually believes that believes that
player $j$ is choosing among the strategies in $\Scal^k_j$ for $k > 1$.
Player $i$'s primary beliefs are characterized by $\Scal^0_{-i}$.  Thus, he
chooses a strategy that minimizes regret with respect to $\Scal^0_{-i}$.
Player $i$'s secondary beliefs are characterized by $\Scal^1_{-i}$.
Intuitively, while player $i$ is almost certain that the other players
are playing a strategy in $\Scal^0_{-i}$, he thinks that there is a small
chance that the other players are using a strategy in $\Scal^1_{-i}$.
Since all the strategies in $\Scal^1_i$ minimize regret with respect to
$\Scal^1_{-i}$, he may as well choose one of the strategies in $\Scal^1_i$
that also minimize regret with respect to $\Scal^1_{-i}$.  More generally,
the strategies in $\Scal^{k+i}_i$ minimize regret with respect to each of 
$\Scal^0_{-i}, \ldots, \Scal^k_{-i}$.
}

\commentout{
But we can provide another intuition for iterated regret minimization,
in terms of \emph{insurance}.  Imagine a broker that is willing to offer
player $i$ a deal: for \$$m$, the broker is willing to guarantee player
$i$ the best payoff, given what the other agents do, if player $i$ will
let the broker play for him.  A fair 
price for this deal is the minimum regret; if the broker plays
an act that minimizes regret for $i$, he is guaranteed to at least break
even.  Once we have such brokers, we can consider a reinsurance
process.  A reinsurance broker can offer the ``first-level'' brokers a
similar deal.  Of course, we can continue the process, giving us
iterated deletion.
}
\subsection{Examples}\label{sec:examples}

We now consider the outcome of iterated regret
minimization in a number of standard games, showing how it 
compares to the strategies recommended by other solution
concepts.  We start by considering what happens if we restrict to pure
strategies, and then consider mixed strategies.

\subsubsection{Pure strategies}\label{sec:pure}

\begin{example} Traveler's Dilemma: \emph{
If $G = ([n],A,\vec{u})$ is the Traveler's Dilemma, then
using the arguments sketched in the introduction, we get that
$\Drm_i^{\infty}(A)= \Drm^2_i(A) = \{100-2p+1
\}$ if $p \geq 50$.  As we mentioned, $(2,2)$ is the only action profile
(and also the only mixed strategy profile) that is rationalizable
(resp., survives iterated deletion of weakly dominated strategies, is a
Nash equilibrium).  On the other hand, no actions or mixed strategies
are strongly dominated, hence all actions (and mixed strategies)
survive iterated deletion of strongly dominated strategies.  Thus,
iterated regret minimization is quite different from all these other
approaches in the Traveler's Dilemma, and gives results that are in
line with experimental observations. 
}
\qed \end{example}

\begin{example} Centipede Game:
\emph{ Another well-known game for which traditional solution concepts provide
an answer that is not consistent with empirical observations is the
\emph{Centipede Game} \cite{rosenthal}.  In the Centipede Game, two
players play for a fixed number $k$ of rounds (known at the outset).
They move in turn; the first player moves in all odd-numbered rounds,
while the second player moves in even-numbered rounds.
At her move, a player can either stop
the game, or continue playing (except at the very last step, when a
player can only stop the game).  For all $t$, player 1 prefers the
stopping outcome in round 
$2t+1$ (when she moves) to the stopping outcome in round $2t+2$;
similarly, for all $t$, player 2 prefers the outcome in round $2t$ (when he
moves) to the
outcome in round $2t+1$.  However, for all $t$, the outcome in round
$t+2$ is better for both players than the outcome in round $t$.  
}

\emph{
Consider two versions of the Centipede Game.  The first has
\emph{exponential} payoffs.  In this case, 
the utility of stopping at odd-numbered rounds $t$
is $(2^t+1,2^t-1)$, while the 
utility of stopping at even-numbered rounds is $(2^{t-1},2^t)$.
Thus, if player 1 stops at round 1, player 1 gets 3 and player 2 gets 2;
if player 1 stops at round 4, then player 1 gets 8 and player 2 gets 16;
if player 1 stops at round 20, the both players get over 1,000,000.
In the version with \emph{linear payoffs with punishment $p>1$}, if $t$
is odd, the payoff is $(t,t-p)$, while if $t$ is even, the payoff is
$(t-p,t)$.
}

\emph{The game can be described as a strategic game where $A_i$ is the set of
strategies for player $i$ in the extensive-form game. 
It is straightforward to show (by backwards induction) that 
the only strategy profiles that survive iterated deletion 
are ones where player 1 stops at the first move and player 2 stops at
the second move.  These are 
also the only rationalizable strategy profiles and the only Nash
equilibria. 
In contrast, in empirical test (which have been done with linear
payoffs), subjects usually cooperate for a 
certain number of rounds, although it is rare for them to cooperate
throughout the whole game \cite{MP92,NT98}.  
As we now show, with iterated regret minimization, we also get
cooperation for a number of rounds (which depends on the penalty); with
exponential payoffs, we get cooperation up to the end of the game.  Our
results suggest some further experimental work, with regard to the
sensitivity of the game to the payoffs.
}

\emph{
Before going on, note that, technically, a strategy in the Centipede Game
must specify what a player does whenever he is called upon to move, including
cases where he is called upon to move after he has already stopped the
game.  Thus, if $t$ is odd and $t+2 < k$, there is more than one
strategy where player 1 stops at round $t$.  For example, there is one
where player 1 also stops at round $t+2$, and another where he continues
at round $t+2$.  However, all the strategies where player 1 first stops
at round $t$ are payoff equivalent for player 1 (and, in particular, are
equivalent with respect to regret minimization).  We use $[t]$ denote
the set of strategies where player 1 stops at round $t$, and similarly
for player 2.
It is easy to see that in the $k$-round Centipede Game with exponential
payoffs, the unique strategy that minimizes regret is to stop at the last
possible round.  On the other hand, with linear payoffs and punishment
$p$, the situation is somewhat similar to the Traveler's Dilemma.
All strategies (actions) for player 1 that stop at or after stage $k-p+1$
have regret $p-1$, which is minimal, but what happens with iterated
deletion depends on whether $k$ and $p$ are even or odd. 
For example, if $k$ and $p$ are
both even, then 
$\Drm_1(A) = 
\{[k-p+1], [k-p+3], \ldots, k-1\}$ and $\Drm_2(A)=\{k-p+2, k-p+4,
\ldots, k\}$.  
Relative to $\Drm(A)$, the strategies in $[k-p+1]$ have regret $p-2$; the
remaining strategies have regret $p-1$.  Thus, $\Drm_1^2(A) =
\Drm_2^\infty(A) = \{[k-p+1]\}$.  
Similarly, $\Drm_2^2(A) = \Drm_2^\infty(A) =\{[k-p+2]\}$.
}

\emph{
If, on the other hand, both $k$ and $p$ are odd, 
$\Drm_1(A) = \{[k-p+1], [k-p+3], \ldots, [k]\}$ and $\Drm_2(A)=\{[k-p],
[k-p+2], \ldots, [k-1]\}$. 
But, here, iteration does not remove any strategies (as here
$k-p+1$ still has regret $p-1$ for player 1, 
and $k-p$ still has regret $p-1$ for player 2.
Thus, $\Drm_1^\infty(A) = \Drm_1(A)$ and $\Drm_2^\infty(A) = \Drm_2(A)$.
}
\qed \end{example}

\commentout{
Things gets more interesting if we consider \emph{finitely repeated Prisoner's
Dilemma}, where the players play Prisoner's Dilemma for some fixed 
number $k$ times (getting the appropriate payoff after each play).  It
is well-known that backwards induction shows that always defecting is
the only Nash equilibrium, the only strategy that is rationalizable, and
the only one that survives iterated deletion of weakly dominated
actions.  And, just as with the Centipede Game, it is not what people do
in practice.  Quite often, we see cooperation.  Indeed, repeated
competitions have shown that \emph{tit-for-tat} and closely related
strategies do well in real competitions \cite{Axelrod}, where
tit-for-tat so the strategy that starts out for cooperating, and at
round $m+1$ plays what the opponent plays at the round $m$ (so, in
particular, punishes the opponent by defecting at round $m+1$ if the
opponent defects at round $m$).  Let $TT^{-}$ be the strategy that is
just like tit-for-tat except that, at the last round, the player
defects.

\begin{proposition} If $a-b < 
$\ttm$ is the unique strategy that minimizes
regret.
\end{proposition}
\begin{proof}  We first show that the regret of $\ttm$ is $(k-1)$.  It is
useful in this discussion to consider $D^k$, the strategy of always
defecting, $c^k$, the strategy of always cooperating, and $\gtrig$, the
\emph{grim trigger} strategy, which  
starts out by cooperating, but once the opponent plays $D$, defects from
then on.  Clearly the regret of $\ttm$ is at least $(k-1)(c-b)$, for if
the opponent plays $C^k$, the best response is clearly $D^k$, which gives a
payoff of $kc$, while playing $\ttm$ against $C^k$ gives a payoff of
$(k-1)b + c$.  Now consider any other strategy.  
}
\begin{example} 
\label{matchingpennies.ex}
Matching pennies: 
\emph{Suppose that 
$A_1 = A_2 = \{a,b\}$, and $u(a,a)=u(b,b)=(80,40),
u(a,b)=u(b,a)=(40,80)$,
and consider the matching pennies game, with payoffs as given in the
table below:}
\begin{table}[htb]
\begin{center}
\begin{tabular}{c |  c c}
& $a$ & $b$\\
\hline
$a$ &$(80,40)$ &$(40,80)$ \\
$b$  &$(40,80)$  &$(80,40)$  \\
\end{tabular}
\end{center}
\end{table}

{\rm Since the players have opposite interest there are no pure strategy
Nash equilibria. Randomizing with equal probability over both actions is
the only Nash equilibria; this is consistent with experimental results
(see e.g., \cite{GoereeHolt}). With regret minimization, both actions
have identical regret for both players, thus using regret minimization
(with respect to pure strategies) both actions are viable.} 

{\rm Consider a variant of this game (called the \emph{asymmetric
matching pennies} \cite{GoereeHoltPalfrey}), where
$u(a,a)=(320,40)$. } 
\begin{table}[htb]
\begin{center}
\begin{tabular}{c |  c c}
& $a$ & $b$\\
\hline
$a$ &$(320,40)$ &$(40,80)$ \\
$b$  &$(40,80)$  &$(80,40)$  \\
\end{tabular}
\end{center}
\end{table}
{\rm 
Here, the unique Nash equilibrium is one where player 1 still randomizes between $a$ and $b$ with equal probability, but player 2 picks $b$ with probability $0.875$ \cite{GoereeHoltPalfrey}. 
Experimental results by Goeree and Holt \citeyear{GoereeHolt} show 
quite different results: player 1 chooses $a$ with probability $0.96$;
player 2, on the other hand, is consistent with the Nash equilibrium and
chooses $b$ with probability $0.86$. In other words, players most of the
time end up with the outcome $(a,b)$. With iterated regret minimization,
we get a qualitatively similar result. It is easy to see that 
in the first round of deletion, $a$ minimizes the regret for player 1,
whereas both $a$ and $b$ minimize the regret for player 2; thus
$\Drm_1^1(A)= {a}$ and $\Drm_2^1(A) = {a,b}$. In the second round of the
iteration, $b$ is the only action that minimize regret for player
2. Thus, $\Drm^2(A) = \Drm^{\infty} = (a,b)$; i.e., $(a,b)$ is the only
strategy profile that survives iterated deletion.} 
\qed
\end{example}

\begin{example} 
\label{coordination.ex}
Coordination games: 
\emph{Suppose that 
$A_1 = A_2 = \{a,b\}$, and $u(a,a)=
(k,k), u(b,b)=(1,1), u(a,b)=u(b,a)=0$, as shown in the table below:
}
\begin{table}[htb]
\begin{center}
\begin{tabular}{c |  c c}
& $a$ & $b$\\
\hline
$a$ &$(k,k)$ &$(0,0)$ \\
$b$  &$(0,0)$  &$(1,1)$  \\
\end{tabular}
\end{center}
\end{table}

\noindent
{\rm Both $(a,a)$ and $(b,b)$ are Nash
equilibria, but $(a,a)$ \emph{Pareto dominates} $(b,b)$ if $k>1$: both
players are better off with the equilibrium $(a,a)$ than with $(b,b)$. 
With regret minimization, we do not have to appeal to Pareto dominance
if we stick to pure strategies.  It is easy to see that 
if $k>1$, $\Drm_1^1(A)= \Drm_2^1(A) = \Drm_1^{\infty}(A)=
\Drm_2^{\infty}(A) = \{a\}$ (yielding regret 1), while 
and if $k=1$, $\Drm_1^1(A)= \Drm_2^1(A) = \Drm_1^{\infty}(A)= \Drm_2^{\infty}(A) =\{a,b\}$.}
%
%
%
%
\qed \end{example}

\begin{example}\label{xam:Bertrand}
Bertrand competition:
\emph{Bertrand competition is a 2-player game where the players can be
viewed as firms producing a homogeneous good.  There is demand
for 100 units of the good at any price up to \$200.  If both firms charge
the same price, each will sell 50 units of the good.  Otherwise, the 
firm that charges the lower price will sell 100 units at that price.
Each firm has a cost of production of 0, as long as it sells for a
positive price, it makes a profit.
It is easy to see that the only Nash equilibria of this game are (0,0)
and (1,1).  But it does not seem so reasonable that firms playing this
game only once will charge \$1, when they could charge up to \$200.
And indeed, experimental evidence shows 
that people in general choose
significantly higher prices \cite{DG00}.}

\emph{Now consider regret.  Suppose that firm 1 charges $n\ge 1$.  If firm 2
charges $m > 1$, then the best response is for firm 1 to charge $m-1$.
If $m > n$, then firm 1's regret is 
$(m-1-n)100$; if $m=n>1$, firm 1's
regret is $(n/2-1)100$; if $m=n=1$, firm 1's regret is 0; and if $m< n$,
firm 1's regret is $(m-1)100$. 
If $m=1$, firm 1's best response is to charge 1, and the regret is
0 if $n=1$ and 100 if $n > 1$.  It follows that 
firm 1's regret is $\max((199-n)100,(n-2)100)$.  Clearly if $n=0$, firm
1's regret is $199 \times 100$ (if firm 2 charges 200).  Thus, firm 1
minimizes regret by playing 100 or 101, and similarly for firm 2.  
A second round of regret minimization, with respect to $\{100,101\}$,
leads to 100 as the unique strategy that results from iterated regret
minimization.   This seems
far closer to what is done in practice in many cases.}
\qed \end{example}

\begin{example} 
\label{xam:bargaining}
The Nash bargaining game \cite{Nash50a}:
\emph{In this 2-player game, each player must choose an integer between
0 and 100.  If player 1 chooses $x$ and player 2 chooses $y$ and $x+y
\le 100$, then player 1 receives $x$ and player 2 receives $y$;
otherwise, both players receive 0.  All strategy profiles of the form
$(x,100-x)$ are Nash equilibria.  The problem is deciding which of these
equilibria should be played.  Nash \citeyear{Nash50a} suggested a number
of axioms, for which it followed that $(50,50)$ was the unique
acceptable strategy profile.} 

\emph{
Using iterated regret minimization leads to the same result, without
requiring additional axioms.  Using arguments much like those used in
the case of Bertrand competition, it is easy to see that the regret of
playing $x$ is $\max(100-x, x-1)$.  If the other player plays $y \le
100-x$, then the best response is $100-y$, and the regret is $100-y-x$.
Clearly, the greatest regret is $100-x$, when $y=0$.  On the other hand,
if the other player plays $y > 100-x$, then the best response is
$100-y$, so the regret is $100-y$.  The greatest regret comes if $y =
100-x+1$, in which case the regret is $x-1$.
It follows that regret is
minimized by playing either 50 or 51.
Iterating regret minimization
with respect to $\{50,51\}$ leaves us with 50.  Thus, $(50,50)$ is the
only strategy profile that survives iterated regret minimization.}

\emph{We have implicitly assumed here that the utility of a payoff
of $x$ is just $x$.  If, more generally, it is $u(x)$, and $u$ is an
increasing function, then the same argument shows that the regret is
$\max(u(100) - u(x), u(x-1) - u(0))$.  Again, there will be either one
value for which the regret is maximized (as there would have been above
if we had taken the total to be 99 instead of 100) or two consecutive
values.  A second round of regret minimization will lead to a single
value; that is, again there will be a single strategy profile of the
form $(x,x)$ that survives iterated regret minimization.  However, $x$
may be such that $2x < 100$ or $2x > 100$.   This can be viewed as a
consequence of the fact that, in a strategy profile that survives
iterated regret minimization, a player is not making a best response to
what the other players are doing.  
}
\qed \end{example}

We next show that that iterated regret and Nash equilibrium agree on
Prisoner's Dilemma.  This follows from a more general observation, that
iterated regret always recommends a dominant action.
A \emph{dominant action} $a$ for player $i$ is one such that
$u_i(a, \vec{b}_{-i}) \ge u_i(a',\vec{b}_{-i})$ for all $a' \in A_i$ and
$\vec{b} \in A$.   We can similarly define a dominant (mixed) strategy.
It is easy to see that dominant actions survive iterated deletion of
weakly and of strongly dominated actions with respect to $A$, and are
rationalizable.  Indeed, if there is a dominant action, the only actions
that survive one round of iterated deletion of weakly dominated
strategies are dominant actions.  Similar observations hold in the case
of mixed strategies.  
The next result shows that iterated regret minimization acts like
iterated deletion of weakly dominated strategies in the presence of
dominant actions and strategies.
\begin{proposition}
\label{dominantaction.prop}
Let $G=(\Play,A,\vec{u})$ be a strategic game. If player $i$ has a
dominant action $a_i$, then 
\begin{itemize}
\item[(a)] $\Drm_i(A) = \Drm_i^\infty(A)$;
\item[(b)] $\Drm_i(A)$ consists of the dominant actions in $G$.
%
\end{itemize}
\end{proposition}
\begin{proof}
Every action that is dominant has regret 0 (which is minimal); furthermore,
only dominant actions have regret 0. It follows that 
$\Drm_i(A)$ consists of the dominant actions for $i$ (if such actions
exist). Since none of these actions will be removed in later deletions,
it follows $\Drm_i(A) = \Drm_i^\infty(A)$.  
\qed
\end{proof}

\begin{example} (Repeated) Prisoner's Dilemma: 
\emph{ Recall that in Prisoner's Dilemma,
players can either cooperate ($c$) or defect ($d$).  They payoffs are $u(d,d)
= (u_1,u_1))$, $u(c,c)=(u_2,u_2)$, $u(d,c) = (u_3,0)$, $u(c,d)=(0,u_3)$, where
$0 < u_1 < u_2 < u_3$ and $u_2 > u_3/2$ (so that alternating between
$(c,d)$ and $(d,c)$ is not as good as always cooperating).  It is well
known (and easy to check) that 
$d$ is the only dominant action for both 1 and 2, so it follows by
Proposition \ref{dominantaction.prop} that traditional solutions
concepts coincide with iterated regret minimization for this game. 
}

\emph{Things get more interesting if we consider repeated
Prisoner's Dilemma. We show that if both players use iterated regret
deletion, they will defect in every round, both in finitely and
infinitely repeated Prisoner's Dilemma.}

\emph{First consider Prisoner's Dilemma repeated $n$ times.
Let $\sigmaad$, the strategy where player 1 always defects, and let $\Strat$
consist of all pure strategies in $n$-round Prisoner's Dilemma.
\begin{lemma}\label{pd}
$\regret_1^{\Strat}(\sigmaad) = (n-1)(u_3 - u_2) + \max(-u_1,u_2-u_3)$.
Moreover, if $\strat$ is a strategy for player 1 where he
plays $c$ before  seeing player 2 play $c$ (i.e., where player 1 either
starts out playing $c$ or plays $c$ at the $k$th for $k>1$ move after seeing
player 2 play $d$ for the first $k-1$ moves), then  
$\regret_1^{\Strat}(\strat) > (n-1)(u_3 - u_2) + \max(-u_1,u_2-u_3)$.
\end{lemma}
}

\emph{It follows from Lemma~\ref{pd} that the only strategies that
remain after one round of deletion are strategies that start out
defecting, and continue to defect as long as the other player defects.
If the players both play such a strategy, they both defect at every round.
Thus, all these strategies that survive one round of deletion
survive iterated deletion.  It follows that 
with iterated regret minimization, we observe defection in every
round of finitely repeated prisoners dilemma.  Essentially the same
argument shows that this is true in
infinitely repeated Prisoner's Dilemma (where payoffs are discounted by
$\delta$, for $0 < \delta < 1$).  By way of contrast, while always
defecting is the only Nash equilibrium in finitely repeated Prisoner's
Dilemma, the Folk Theorem shows that for all $p$ with $0 \le p \le 1$,
if $\delta$ is sufficiently close to 1, 
there is an equilibrium in which $p$ is the fraction of times that both
players cooperate.  Thus, with Nash equilibrium, there is a
discontinuity between the behavior in finitely and infinitely repeated
Prisoner's Dilemma that does not occur with regret minimization.
Nevertheless, intuition suggests that there should be a way to justify
cooperation using regret minimization, just as in the case of the
Centipede Game.  This is indeed the case, as we show in
Section~\ref{sec:prior}.  }
\qed \end{example}

\begin{example}
Hawk-Dove:%
\fullv{\footnote{This game is sometimes known as \emph{Chicken}. }}
\emph{ 
In this game, $A_1 = A_2 = \{d,h\}$;
a player can choose to be a dove ($d$) or a hawk ($h$).  The payoffs are
something like Prisoner's Dilemma (with $h$ playing the role of
``defect''), but the roles of $a$ and 0 are switched.  Thus, we have 
$u(d,d) = (b,b), u(d,h) = (a,c), u(h,d)=(c,a), u(h,h)=(0,0)$, 
where again $0 < a < b < c$.  This switch of the role of $a$ and 0
results in the game having two Nash equilibria: $(h,d)$ and $(d,h)$. 
But $d$ is the only action that minimizes regret (yielding regret $c-b$).
Thus, $\Drm(A) = \Drm^\infty(A) = (d,d)$.
\commentout{
For mixed
strategies, consider a strategy for player 1 that puts probability $p$ on $h$. 
If player 2 plays $h$, then the best response is $d$, and the regret of
this strategy is $b - (1-p)b = pb$; similarly, if player 2 played $d$,
the best response is $h$ and the regret of this strategy is $c - (pc +
(1-p)a) = (1-p)(c-a)$.  It follows (by Proposition \ref{onlypure.prop}) that this strategy has regret is $\max(pb,
(1-p)(c-a))$, which is minimized when $pb = (1-p)(c-a)$; i.e., when 
$p = (c-a)/(b+c-a)$.     
}
}
\qed \end{example}

\begin{example} 
\emph{In all the examples considered thus far, there were no strongly
dominated strategies, so all strategies survived iterated deletion of
strongly dominated strategies.  The game described below
shows that the
strategies that survive iterated deletion of strongly dominated
strategies can be disjoint from those that survive iterated regret
minimization.  }

\begin{figure}[htb]   
\begin{center}   
\begin{tabular}{c |  c c}
& $x$ & $y$\\
\hline
 $a$& $(0,100)$& $(0,0)$ \\ 
 $b$&$ (1,0)$&$(1,1)$\\ 
%

%
\end{tabular}\\   
\end{center}   
\end{figure}   
\emph{First consider iterated deletion of strongly dominated strategies.
For player 1, $b$ strongly dominates $a$.
Once $a$ is deleted, $y$ strongly dominates $x$ for player 2. Thus,
iterated deletion of strictly dominated strategies
leads to the unique strategy profile $(b,y)$.  
Now consider regret minimization. The regret of $x$ is
less than 
that of $y$, while the regret of $b$ is less than that of
$a$.  Thus, iterated regret minimization leads to $(b,x)$.}
\qed \end{example}

%

%
%
%
%
%
%
%
Note that in the examples above, the deletion process converges after two
steps.
We can construct examples of games where we need $\max(|A|_1 - 1, \ldots,
|A|_n - 1)$ deletion steps.  
The following example shows this in the case that $n=2$ and $|A_1| =
|A_2|$, and arguably illustrates some problems with iterated regret
minimization.  

\begin{example}
\emph{Consider a symmetric 2-player game with 
$A_1 = A_2 = \{a_1, \ldots, a_n\}$. 
If both players play $a_k$, then the payoff for each one is $k$.  For $k
> 1$, if one player plays $a_k$ and the other plays $a_{k-1}$, then the player
who plays $a_k$ get $-2$, and the other gets 0.  In all other cases,
both players get payoff of 0.  The $ij$ entry of the following matrix
describes player $h$'s payoff if $h$ plays $a_i$ and player $2-h$ plays
$a_j$: }

$$\left[
\begin{array}{cccccc}
1  & 0  & \ldots& 0 &  0\\
-2 & 2  & \ldots & 0 & 0 \\
0 & -2 & 3 & \ldots  & 0\\
\vdots & \vdots  &  \vdots & \vdots     & \vdots \\
0  & \ldots & 0  & -2 & n
\end{array}
\right] $$

\emph{Note that $a_n$ has regret $n+1$ (if, for example player 1 plays
$a_n$ and player 2 plays $a_{n-1}$, player 1 could have gotten $n+1$
more by playing $a_{n-1}$; all other actions have regret $n$.  Thus,
only $a_n$ is eliminated in the first round.  Similar considerations
show that we eliminate $a_{n-1}$ in the next round, then $a_{n-2}$, and
so on. Thus, $a_1$ is the only pure strategy that survives iterated
regret minimization.}

\emph{Note that $(a_k,a_k)$ is a Nash equilibrium for all $k$.  Thus, the
strategy that survives iterated regret minimization is the one that is
Pareto dominated by all other Nash equilibria.  We get a similar result
if we modify the payoffs so that if both players play $a_k$, then they
both get $2^k$, while if one player plays $a_k$ and the other plays
$a_{k-1}$, then the one that plays $a_k$ get $-2 - 2^{k-1}$, and the
other gets 0.  Suppose that $n=20$ in the latter game.  Would players
really accept a payoff of 2 when they could get a payoff of over
1,000,000 if they could coordinate on $a_{20}$?  Perhaps they would not
play $a_{20}$ if they were concerned about the loss they would face if
the other player played $a_{19}$.}

{\em The following variant of a \emph{generalized coordination game}
demonstrates the same effect 
even without 
iteration.
\begin{center}
\begin{tabular}{c |  c c}
& $a$ & $b$\\
\hline
$a$ &$(1,1)$ &$(0,-10)$ \\
$b$  &$(-10,0)$  &$(10,10)$  \\
\end{tabular}
\end{center}
Clearly $(b,b)$ is the Pareto optimal Nash equilibrium, but playing $b$
has regret 11, whereas $a$ has regret only 10; thus $(a,a)$ is the only
profile that minimizes regret. 
Note, however, that $(a,a)$ is the \emph{risk dominant} Nash equilibrium.
(Recall that in a \emph{generalized coordination game}---a 2-player,
2-action game where $u_1(a,a)>u_1(b,a)$, $u_1(b,b)>u_1(a,b)$, $u_2(a,a)
> u_2(a,b)$, and $u_2(b,b) > u_2(b,a)$---%
the Nash   
equilibrium $(a,a)$ is risk dominant if the product of the ``deviation
losses'' for $(b,b)$ (i.e., $(u_1(b,a)-u_1(b,b))(u_2(a,b)-u_2(b,b))$) is
higher than the product of the deviation losses for $(a,a)$.) In the game
above, the product of the deviation losses for $(b,b)$ is $100=(0-10)(0-10)$,
while 
the product of the deviation losses for $(a,a)$ is $121 = (-10-1)(-10-1)$;
thus, $(a,a)$ is risk dominant. In fact, in every generalized
coordination game, the product of the deviation losses for $(x,x)$ is
$regret_1(x) regret_2(x)$, so if the game is symmetric (i.e., if
$u_1(x,y)=u_2(y,x)$, which implies that $regret_1(x) = regret_2(x)$), the
risk dominant Nash equilibrium is the only action profile that minimizes
regret.  (It is easy to see that this is no longer the case if the game is
asymmetric.) 
}
\qed \end{example}

\subsubsection{Mixed strategies}\label{sec:mixed}
Applying regret in the presence of mixed strategies can lead to quite
different results than if we consider only pure strategies.  
We first generalize
Proposition~\ref{dominantaction.prop}, 
to show that 
if there is a dominant strategy 
that happens to be pure (as is the case, for example, in Prisoner's
Dilemma), nothing changes in the presence of mixed strategies. 
But in general, things can change significantly.
\begin{proposition}
\label{dominantactionmixed.prop}
Let $G=(\Play,A,\vec{u})$ be a strategic game. If player $i$ has a
dominant action $a_i$, then $\Drm_i(\Smixed) = \Drm_i^\infty(\Smixed) =
\Delta(\Drm_i(A))$.  
%
%
%
%
%
\end{proposition}
\begin{proof}
The argument is similar to the proof of Proposition
\ref{dominantaction.prop} and is left to 
the reader. 
\qed
\end{proof}

To understand what happens in general, we first 
shows that we need to consider regret relative to only pure strategies
when minimizing regret at the first step.
(The proof is relegated to the Appendix.)

\begin{proposition}
\label{onlypure.prop}
Let $G=(\Play,A,\vec{u})$ be a strategic game and 
let $\sigma_i$ be a mixed strategy for player $i$.
Then $\regret_{i}^{\Sigma}(\sigma_i) = \max_{\vec{a}_{-i} \in A_{-i}}
\regret^{\Scal_{i}}(\sigma_i \mid \vec{a}_{-i})$.
\end{proposition}

\begin{example} Roshambo (Rock-Paper-Scissors): 
\emph{In the rock-paper-scissors game,
$A_1 = A_2 =
\{r,s,p\}$;
rock ($r$) beats scissors ($s$), scissors beats paper ($p$), and paper
beats rock; $u(a,b)$ is $(2,0)$ if $a$ beats $b$, $(0,2)$ if $b$ beats
$a$, and $(1,1)$ if $a=b$.
If we stick with pure strategies, by symmetry, we have $\Drm_1(A) =
\Drm_2(A) = \Drm^\infty_1(A) = \Drm^\infty_2(A) = \{r,s,p\}$.
If we move to mixed strategies, 
it follows by Proposition \ref{onlypure.prop} that
picking $r$, $s$, and $p$ each with probability $1/3$ is the only
strategy that minimizes  
regret. (This is also the only Nash equilibrium.)
%
}
\qed \end{example}

\begin{example} 
\label{mp1.ex}
Matching pennies with mixed strategies: 
\emph{Consider again the matching pennies game, where
$A_1 = A_2 = \{a,b\}$, and $u(a,a)=u(b,b)=(80,40), u(a,b)=u(b,a)=(40,80)$.
Recall that $\Drm_1^{\infty}(A)= \Drm_2^{\infty}(A)= \{(a,b)\}$.
Consider a mixed strategy that puts weight $p$ on $a$.
By Proposition \ref{onlypure.prop}, the regret of this strategy is 
$\max (40(1-p), 40p)$, which is
minimized when $p=\frac{1}{2}$ (yielding regret $20$).
Thus randomizing with equal probability over $a,b$ is the only strategy
that minimizes regret; it is also the only Nash equilibrium. 
%
%
But now consider the asymmetric matching pennies game, where
$u(a,a)=(320,40)$. 
Recall that $\Drm^{infty}(A) = (a,b)$. 
Since the utilities have not changed for player 2, 
it is still the case that
$1/2a + 1/2b$ is the only strategy that minimizes regret for player 2.
On the other hand, if 
by Proposition \ref{onlypure.prop}, the regret of the strategy for
player 1 that puts weight $p$ on $a$ is
$\max (280(1-p), 40p)$, which is
minimized when $p=.875$
Thus $\Drm^{infty}(\Sigma)= (.875a+.225b, 0.5a+0.5b)$.}
\qed
\end{example}

\begin{example} 
\label{coordination1.ex}
Coordination games with mixed strategies: 
\emph{Consider again the coordination game where
$A_1 = A_2 = \{a,b\}$ and $u(a,a)=
(k,k), u(b,b)=(1,1), u(a,b)=u(b,a)=0$. 
Recall that if $k>1$, then $\Drm_1^{\infty}(A)= \Drm_2^{\infty}(A)= \{(a)\}$,
while if $k=1$, then $\Drm_1^{\infty}(A)= \Drm_2^{\infty}(A) = \{a,b\}$.
Things change if we consider mixed strategies.
Consider a mixed strategy that puts weight $p$ on $b$.
By Proposition \ref{onlypure.prop}, the regret of this strategy is 
$\max (kp,1-p)$ which is
minimized when $p=\frac{1}{k+1}$ (yielding regret $\frac{k}{k+1}$).
Thus, mixed strategies that minimize regret can put 
positive weight on actions that have sub-optimal regret.}
\qed \end{example}

\begin{example} Traveler's Dilemma with mixed strategies: 
\emph{As we saw earlier, each of the choices 96--100 has regret 3 relative to
other pure strategies.  It turns out that there are mixed strategies
that have regret less than 3.  Consider the mixed strategy that puts
probability 1/2 on 100, 1/4 on 99, and decreases exponentially, putting
probability $1/2^{98}$ on both 3 and  2.  Call this mixed strategy
$\sigma$.  Let $\Sigma$ consist of all the mixed strategies for
Traveler's Dilemma.}

\emph{\begin{lemma}\label{lem:mixedregret}  $\regret_1^{\Sigma}(\sigma) < 3$.
\end{lemma}
}

\emph{The proof of Lemma~\ref{lem:mixedregret} shows that
$\regret_1^{\Sigma}(\sigma)$ is not much less than 3; it is 
roughly $3 \times (1 - 1/2^{101-k})$.  Nor is $\sigma$ 
the strategy that minimizes regret.  For example, it follows from the
proof of Lemma~\ref{lem:mixedregret} that we can
do better by using a strategy that puts probability 1/2 on 98 and
decreases exponentially from there.  While we have not computed the exact
strategy that minimizes (or strategies that minimize) regret---the
computation is nontrivial and does not add much insight---we can make
two useful observations: 
\begin{itemize}
\item The mixed strategy that minimizes regret places probability at
most $3/(99-k)$ on the pure strategies $k$ or less.  For suppose player
places probability $\alpha$ on the pure strategies $k$ or less.
If player 2 plays
100, player 1 could have gotten 101 by playing 99, and gets at most
$k+2$ by playing $k$ or less.  Thus, the regret is at least
$(99-k)\alpha$, which is at least 3 if $\alpha \ge 3/(99-k)$.  Thus, for
example, the probability that 90 or less is played is at most $1/3$.
\item The strategy that minimizes regret has regret greater than 2.9.  To
see this, note that the strategy can put probability at most $3/97$ on 2 and
at most $3/96$ on 3.
This means that the regret relative to 3 is at least 
$$(1 - 3/96 - 3/97)3 + 3/96 = 
3 - 6/96 - 3/97 > 2.9).$$
\end{itemize}
}

\emph{
The fact that it is hard to compute the exact strategy the minimizes regret
suggests that people are unlikely to be using it.  On the other hand, it
is easy to compute that the optimal strategy puts high weight on actions
in the high 90's.  In retrospect, it is also not surprising that one can
come close to minimizing regret by putting some weight on (almost) all 
actions.  This was also the case in Example~\ref{coordination1.ex}; as
we observed there, we can sometimes do better by putting some weight
even on actions that do not minimize regret.%
\footnote{It is not necessarily the case that the support of the optimal
strategy consists of all actions, or even all undominated actions.  For
example, consider a coordination game with three actions $a$, $b$, and
$c$, where $u(a,a) = u(b,b) = k$, $u(c,c) = 1$, and $u(x,y) = 0$ if $x
\ne y$.  If $k > 2$, then the strategy that minimizes regret places
probability 1/2 on each of $a$ and $b$, and probability 0 on $c$.}
Interestingly, the distribution of strategies observed by 
Becker, Carter, and Naeve \citeyear{BCN05} is qualitatively similar to
the distribution induced by a mixed strategy that is close to optimal.  
If everyone in the population was playing a mixed strategy that was
close to optimal in terms of minimizing regret, we would expect to see
something close to the observed distribution. 
}
\qed \end{example}

The phenomena observed in the previous two examples apply to all the
other examples considered in Section~\ref{sec:pure}.  For example, in
the Bertrand competition, while the pure strategies of least regret (100
and 101) have regret 9,900, there are mixed strategies with regret less
than 7,900 (e.g., by putting probability $1/5$ on each of 80, 100, 120,
140, and 160).  We can do somewhat better than this, but not much.
Moreover, we believe that in both Traveler's Dilemma and in Bertrand
Competition there is a unique mixed strategy that 
minimizes regret, so that one round of deletion will suffice.
This is not true in
general, as the follow example shows.

\begin{example}\label{xam:mixedmultiround}
\emph{Consider a 2-player symmetric game where $A_1 = A_2 = \{a_{mk}: 
m= 1, \ldots, n, k =1,2\}$.  Define 
$$
u_1(a_{ij},a_{kl}) = 
\left\{
\begin{array}{ll}
-3^{max(i,k)}       &\mbox{if $i \neq k$}\\
0          &\mbox{if $i=k$, $j=l$}\\
-3^{i+1}   &\mbox{if $i=k$, $j \ne l$}
\end{array}\right.
$$
}
\emph{Let $\Sigma$ consists of all mixed strategies in this game. 
We claim that, for every strategy $\sigma$ for player 1,
$\regret_1^{\Sigma}(\sigma) \ge 3^n$, and similarly for player 2.
To see this, consider a mixed strategy of the form
$\sum_{ij} p_{ij}a_{ij}$.  The best response to $a_{nj}$ is $a_{nj}$, which
gives a payoff of 0.  Thus, the regret of this strategy relative to
$a_{n1}$ is $3^n(\sum_{ij, i \ne n} p_{ij} + 3p_{n2})$.
Similarly, the regret relative to $a_{n3}$ is
$3^n(\sum_{ij, i\ne n} p_{ij} + 3p_{n1})$. 
Thus, the sum of the regrets relative to
$a_{n1}$ and 
$a_{n2}$ is $3^n(2 + p_{n1} + p_{n_2})$.  It follows that the regret relative
to one of $a_{n1}$ and $a_{n2}$ is at least $3^n$.   It also easily
follows that  every convex combination of strategies  $a_{ij}$ with $i < n$
has regret exactly 
$3^n$ (the regret relative to $a_{n1}$ is $3^n$, and the regret relative
to every other strategy is no worse).  Moreover, every strategy that puts
positive weight on $a_{n1}$ or $a_{n2}$ has regret greater than $3^n$.  Thus,
at the first step, we eliminate all and only strategies that put
positive weight on $a_{n1}$ or $a_{n2}$.  An easy induction shows that at
the $k$th step we eliminate all and only strategies that put positive
weight on $a_{(n-k+1)1}$ or $a_{(n-k+1)2}$.  After $n-1$ steps of
iterated deletion, the only strategies that are convex combinations of
$a_{11}$ and $a_{12}$.  One more step of deletion leaves us $1/2 a_{11}
+ 1/2 a_{12}$.%
\footnote{In this example, at every step but the last step, the set of
strategies that remain consist of all convex combinations of a subset of
 pure strategies.  But this is not necessarily the case.  If we replace
$3^k$ by $2^k$ in all the utilities above, then we do not eliminate all
strategies that put positive weight on $a_{n1}$ or $a_{n2}$; in particular,
we do not eliminate strategies that put the \emph{same} weight on $a_{n1}$ and
$a_{n2}$ (i.e., where $p_{n1}=p_{n2}$).}
}  
\end{example} 

In the case of pure strategies, it is immediate that there cannot be
more than $|A_1| + \cdots + |A_n|$ rounds of deletion, although we do
not have an example that requires more than $\max(|A_1|, \ldots, |A_n|)$
rounds.  Example~\ref{xam:mixedmultiround} shows that $\max(|A_1|,
\ldots, |A_n|)$ may be required with mixed strategies, but all that
follows from Theorem~\ref{nonempypure.prop} is that the deletion process
converges after at most countably many steps.  
We conjecture
that, in fact, the process converges after at most $\max(|A_1|, \ldots,
|A_n|)$ steps, both with pure and mixed strategies, but we have not
proved this.

\subsection{Iterated regret minimization with prior
beliefs}\label{sec:prior} 
We have assumed that we start the deletion process with all pure (resp.,
mixed) strategy profiles.  
Moreover, we have assumed that, at all stages in the deletion process
(and, in particular, at the beginning), the set of strategies that the
agents consider possible is the same for all agents.
More generally, we 
could allow each agent $i$ could start 
a stage in 
the deletion process with a set
$\Sigma^i$ of strategy profiles.  Intuitively, the strategies in
$\Sigma^i_j$ are the the only strategies that $i$ is considering for
$j$.  For $j \ne i$, it is perhaps most natural to think of
$\Sigma^i_j$ as representing $i$'s beliefs about what strategies $j$
will use; however, it may also make sense to interpret $\Sigma^i_j$ as a
representative set of $j$'s strategies from $i$'s point of view, or as
the only ones that $i$ is considering but it is too complicated to
consider them all (see below). 
For $i = j$, $\Sigma_i^i$ is the set
of strategies that $i$ is still considering using; thus, $i$
essentially ignores all strategies other than those in $\Sigma_i^i$.  
When we do regret minimization with respect to a single set $\Scal$ of
strategy profiles (as we do in the definition of iterated regret
minimization), we are implicitly assuming that the players have 
common beliefs.    

The changes required to deal with this generalization are
straightforward: each agent simply applies the standard regret
minimization operator to his set of strategy profiles.  More formally,
the generalized regret minimization $\Drm'$ takes as an argument a tuple
$(\Pi_1,\ldots,\Pi_n)$ of strategy profiles and returns such a tuple; we
define $\Drm(\Pi_1,\ldots,\Pi_n) = (\Drm(\Pi_1), \ldots, \Drm(\Pi_n))$.%
\footnote{As we hinted in Section~\ref{sec:characterization}, 
an epistemic justification of this more general notion of regret
minimization would require a more general notion of lexicographic
beliefs, where each player has a separate sequence of beliefs.}

\begin{example} Repeated Prisoner's Dilemma with prior beliefs: \emph{
The role of prior beliefs is particularly well illustrated in Repeated
Prisoner's Dilemma.  In the proof of Lemma~\ref{pd}, to show that the regret
of a strategy like Tit for Tat is greater $(n-1)(u_3 - u_2) +
\max(-u_1,u_2-u_3)$, it is 
necessary to consider a strategy where player 2 starts out by
defecting, and then cooperates as long as player 1 defects.
This seems like an extremely unreasonable strategy for player 2 to use!
Given that there are $2^{2^n-1}$ pure strategies for each player in
$n$-round Prisoner's Dilemma, and computing the regret of each one can
be rather complicated, it is reasonable for the players to focus on a
much more limited set of strategies.  
Suppose that 
each player believes that the other player is using a strategy where
plays Tit for Tat for some number $k$ of rounds, and then defects from
then on, for some $k$.  Call this strategy $\strat_k$.  (So, in
particular, $\strat_0 = \sigmaad$ and $s_n$ is Tit for Tat.)   
Let $\Strat^*_i $ consist of all the strategies $\strat_k$ for player $i$; 
let $\Strat^+_{2-i}$ be any set of strategies for player $2-i$ that
includes $\Strat^*_{2-i}$.  
It is easy to see that 
the best response to $\strat_0$ is $\strat_0$,
and the best response to $\strat_k$ for $k > 1$ is $\strat_{k-1}$ (i.e.,
you are best off defecting just before the other player starts to
defect).  
Thus, $$\regret_i^{\Strat^+_{i} \times \Strat^*_{2-i}}(\strat_k \mid
\strat_l) = 
\left \{\begin{array}{ll}(l-k-1)(u_2 - u_1) &\mbox{if $k  < l$}\\
u_3 + u_1 - 2u_2 &\mbox{if $k=l > 0$}\\
u_3 + u_1 - u_2 &\mbox{if $k>l > 0$}\\
u_1 &\mbox{if $k>l = 0$}\\
0  &\mbox{if $k =l = 0$.}
\end{array}\right.
$$
It follows that $$
\regret_i^{\Strat^+_{i} \times \Strat^*_{2-i}}(\strat_k) = \left\{
\begin{array}{ll}
\max((n-k-1)(u_2 - u_1), u_3 + u_1 - u_2) &\mbox{if $k \ge 2$}\\
\max((n-1)(u_2 - u_1), u_3 + u_1 - 2u_2, u1) &\mbox{if $k = 1$}\\
n(u_2 - u_1)                       &\mbox{if $k = 0$}
\end{array}\right.
$$
Intuitively, if player 1 plays $\strat_k$ and player 2 is playing a
strategy in $\Strat_2^*$, then player 1's regret is maximized
if player 2 plays either $\strat_n$ (in which case 1 would have been
better off by continuing to cooperate longer) or if player 2 plays
$\strat_{k-1}$ (assuming that $k > 1$), in which case 1 would have been
better off by defecting earlier.  Thus, the strategy that minimizes
player 1's regret is either $\strat_{n-1}$, $\strat_1$, or $\strat_0$.
(This is true whatever strategies player 1 is considering for himself,
as long as it includes these strategies.)
If $n$ is 
sufficiently large, then it will be $\strat_{n-1}$.  This seems
intuitively reasonable.  In the long run, a long stretch of cooperation
pays off, and minimizes regret.  Moreover, it is not hard to show that
allowing mixtures over $\strat_0, \ldots, \strat_n$ makes no
difference; for large $n$, $\strat_{n-1}$ is still the unique strategy
that minimizes regret.} 

\emph{To summarize, if each player $i$ believes that the other player $2-i$ is
playing a strategy in $\Strat_{2-i}^*$---a reasonable set of strategies
to consider---then we get a strategy that looks much more like what
people do in practice.
}
\qed
\end{example}

Thinking in terms of beliefs makes it easy to relate iterated regret to
other notions of equilibrium.  Suppose that there exists a strategy
profile $\vec{\sigma}$ such that player $i$'s beliefs have the form 
$\Smixed_i \times \{\sigma_{-i}\}$.  That is, player $i$ believes that
each of the other players are playing their component of $\vec{\sigma}$,
and there are no constraints on his choice of strategy.  Then it is easy
to see that the strategies that minimize player $i$'s regret with
respect to these beliefs are just the best responses to $\sigma_{-i}$.   
In particular, if $\vec{\sigma}$ is a Nash equilibrium, then 
$(\vec{\sigma}, \ldots,\vec{\sigma}) \in \Drm'(\Smixed_1 \times
\sigma_{-1}, \ldots, \Smixed_n \times \sigma_{-n})$. 
The key point here is that if the agent's beliefs are represented by a
``small'' set, then the agent makes a best response in the standard
sense by minimizing regret; minimizing regret with respect to a ``large''
belief set looks more like traditional regret minimization.

\section{Iterated Regret Minimization in Bayesian
Games}\label{sec:Bayesian} 

\emph{Bayesian games} are a well-known generalization of strategic
games, where each agent is assumed to have a
characteristic or some private information not known to the other
players.  This is modeled by assuming each player has a
\emph{type}.  Typically it is assumed that that there is a commonly
known probability over the set of possible type profiles.  Thus, a
Bayesian game is tuple $(\Play, A, \vec{u}, T, \pi)$, where, as before,
$\Play$ is the set of players, $A$ is the set of action profiles,
$\vec{u}$ is the profile of utility functions, $T = T_1 \times \ldots
\times T_n$ is the set of type profiles (where $T_i$ represents the set
of possible types for player $i$), and $\pi$ is a probability measure on
$T$.  A player's utility can depend, not just on the action profile, but
on the type profile.  Thus, $u_i :  A \times T \rightarrow \IR$.
For simplicity, we assume that $\Pr(t_i) > 0$ for all types $t_i \in
T_i$ and $i = 1, \ldots, n$ (where $t_i$ is an abbreviation of
$\{\vec{t}': t'_i = t_i\}$).  

A strategy for player $i$ in a Bayesian game in a function from player
$i$'s type to an action in $A_i$; that is, what a player does will in
general depends on his type.  For a pure strategy profile $\vec{\sigma}$, 
player $i$'s expected utility is 
$$U_i(\vec{\sigma}) = \sum_{\vec{t} \in T} \pi(\vec{t})
u_i(\sigma_1(t_1), \ldots, \sigma_n(t_n)).$$
Player $i$'s expected utility with a mixed strategy profile
$\vec{\sigma}$ is computed by computing the expectation with respect to
the probability on pure strategy profiles induced by $\vec{\sigma}$.  
Given these definitions, a Nash equilibrium in a Bayesian game is
defined in the same way as a Nash equilibrium in a strategic game.

There are some subtleties involved in doing iterated deletion in Bayesian
games.  
Roughly speaking, we need to relativize all the previous
definitions so that they take types into account.  For ease of
exposition, we give the definitions for pure strategies; the
modifications to deal with mixed strategies are straightforward and left
to the reader.

As before, suppose that $\Scal = \Scal_1 \times \ldots \times \Scal_n$.
Moreover, suppose that, for each player $i$, $\Scal_i$ is also a
crossproduct, in that, for each type $t \in T_i$, there exists a set of
actions $A(t)\in A_i$ such that $\Scal_i$ consists of all strategies
$\sigma$ such that $\sigma(t) \in A(t)$ for all $t \in T_i$.  
For $\vec{a}_{-i} \in \Scal_{-i}$ and type vector $\vec{t}$, let
$u^{\Scal_i}_i(\vec{a}_{-i},\vec{t}) = \max_{a_i\in \Scal}
u_i(a_i,\vec{a}_{-i},\vec{t})$. 
For $a_i \in \Scal_i$, $\vec{a}_{-i} \in \Scal_{-i}$, and type profile
$\vec{t}$, let
the {\em regret\/} of $a_i$ for player $i$ given $\vec{a}_{-i}$ and type
profile $\vec{t}$ relative
to $\Scal_i$, denoted
$\regret^{\Scal_i}_{i}(a_i \mid \vec{a}_{-i}, \vec{t})$, be 
$u_i^{\Scal_{i}}(\vec{a}_{-i},\vec{t}) - u_i(a_i,\vec{a}_{-i},\vec{t})$.
Let $\regret_{i}^{\Scal}(a_i \mid \vec{t}) = \max_{\vec{a}_{-i} \in
\Scal_{-i}(\vec{t}_{-i})} \regret^{\Scal_{i}}(a_i \mid
\vec{a}_{-i},\vec{t})$ denote the the maximum regret of 
player $i$ given $\vec{t}$.
The \emph{expected regret of $a_i$ given $t_i$ and $\Scal_{-i}$} 
is $E[\regret_i^{\Scal_i}(a_i \mid t_i)] =
\sum_{\vec{t} \in T} \Pr(\vec{t} \mid t_i) \regret_i^{\Scal_i}
(a_i \mid \vec{t})$.
Let $\minregret^{\Scal_i}(t_i) = \min_{a_{i} \in \Scal_i(t_i) }
E[\regret_i^{\Scal}(a_i \mid t_i)$.  
We delete all those strategies that do not give an action that minimizes
expected regret for each type.  Thus, let
$\Drm_i(\Scal_i) = \{\sigma \in \Scal_i: \regret_i^{\Scal_i}(\sigma(t_i)
\mid t_i) = \minregret^{\Scal_i}(t_i)\}$.  As usual, we take
$\Drm(\Scal) = \Drm_1(\Scal_1) \times \ldots \times \Drm_n(\Scal_n)$.
Having defined the deletion operator, we can apply iterated
deletion as before.  

\begin{example} Second-Price Auction: \emph{A second-price auction can
be modeled as a Bayesian game, where  a
player's type is his valuation of the product being auctioned.  His
possible actions are bids.  The player with
the highest bid wins the auction, but pays only what the second-highest
(For simplicity, we assume that in the event of a tie, 
the lower-numbered player wins the auction.)  
If he bids $b$ and has
valuation (type) $v$, his utility is $b-v$; if he does not win the
auction, his utility is 0.  As is well known, in a second-price auction,
the strategy where each player bids his type is weakly dominant; hence,
this strategy survives iterated regret minimization.  No other strategy
can give a higher payoff, no matter what the  type profile is.}
\qed \end{example}

\begin{example} First-Price Auction: \emph{In a first-price auction, the
player with the highest bid wins the auction, but pays his actual bid.   
Assume, for simplicity, 
that bids are natural numbers, 
that the lower-numbered player wins the auction in the event of a tie,
that all valuations are even, 
and that the product is sold only if some
player bids above 0. 
If a player's valuation is $v$, then bidding $v'$ has 
regret $max(v'-1,v-v'-1)$.
To see this, consider player $i$.
Suppose that the highest bid of the other agents is $v''$, and that
the highest-numbered agent that bids $v''$ is agent $j$.  If $v'' < v'$
or $v'' = v'$ and $i < j$, then $i$ wins the bid.  He may have done
better by bidding lower, but the lowest he can bid and still win is 1,
so his maximum regret in this case is $v'-1$ (which occurs, for example,
if $v'' = 0$).  On the other hand, if $v'' > v'$ or $v'' = v$ and $j <
i$, then $i$ does not win the auction.  He feels regret if he could have
won the auction and still paid at most $v$.  To win, if $j < i$, he must
bid $v'' + 1$, in which case his regret is $v - v'' - 1$.  In this case,
$v''$ can be as small as $v'$, so his regret is at most $v - v' - 1$.  
If $j > i$, then he must only bid $v''$ to win, so his regret is $v -
v''$, but $v'' \ge v'+1$, so his regret is again at most $v - v' - 1$.
It follows that bidding $v'=v/2$ is the unique action that minimizes
regret (yielding a regret of $v/2-1$). 
}
\qed \end{example}

\section{Mechanism Design using Regret Minimization}
\label{mechdesign.sec}
In this section,  we show how using regret minimization as the solution
concept can help to construct efficient mechanisms.  
We consider mechanisms where an agent truthfully reporting his
type is the unique strategy that minimizes regret, and focus on
prior-free mechanisms (i.e., mechanisms that do not depend on the type
distribution of players); we call 
these \emph{regret-minimizing truthful mechanisms}. 
Additionally, we focus on ex-post individually-rational (IR)
mechanisms.\footnote{Recall that a mechanism is ex-post individually
rational if a player's utility of participating is no less than that of
not participating, no matter what the outcome is.} 
As the example below shows, regret-minimizing truthful mechanisms can do
significantly better than dominant-strategy truthful mechanism. 

\begin{example} 
\label{combauction.ex}
Maximizing revenue in combinatorial auctions:  
\emph{In a combinatorial auction, there is a set of $m$ indivisible items
that are concurrently auctioned to $n$ bidders. The bidders can bid on
bundles of items (and have a valuation for each such 
bundle); the auctioneer allocates the items to the bidders. The
standard VCG mechanism 
is known to maximize social welfare (i.e., the allocation by the
auctioneer maximizes the sum of the valuations of the bidders of the
items they are assigned), but might yield poor revenue for the
seller. Designing combinatorial auctions that provide good revenue
guarantees for the seller is a recognized open problem.
By using regret minimization as the solution concept, we can
provide a straightforward solution.} 

\emph{Consider a combinatorial
first-price auction: that is, the auctioneer determines the allocation that
maximizes its revenue (based on the bidders' bids), and the winning
bidders pay what they bid. Using the same argument as for the the case
of a single-item first-price auction, it follows that if a bidder's
valuation of a bundle is $v$ (where $v$ is an even number), bidding
$v/2$ is the unique bid on that bundle that minimizes his regret. Thus,
in a combinatorial first-price auction, the seller is guaranteed to
receive $\MSW/2$, where $\MSW$ denotes the maximal social welfare,
that is, the maximum possible sum of the bidders' valuation for an
allocation.  Clearly $\MSW$ is the most that the seller can receive,
since a rational bidder would not bid more than his valuation.
To
additionally get a truthful auction with the same guarantee, change the
mechanism so that the winning bidder pay $b/2$ if he bids $b$;
it immediately follows that a bidder with valuation $v$ for a bundle should
bid $v$. 
(The mechanism is also trivially IR as players never pay more than
their valuation.)
This should be contrasted with the fact that there is no
dominant-strategy implementation (i.e., no mechanism where bidding the
valuation maximizes utility no matter what the other player bid)  
that guarantees even a positive fraction of $\MSW$ as revenue.  
In fact, as we now show, to guarantee a fraction $r$ of $\MSW$ the
minimum regret needs to be ``large''.  By way of contrast, dominant
strategies have regret 0.}
\begin{lemma} \label{combauc.clm1} An efficient, IR, regret-minimizing
truthful mechanism that guarantees 
the seller a fraction $r$ of $\MSW$ as revenue has a minimum regret of
at least $r\MSW-1$. (In particular, a mechanism that guarantees a revenue
of $\MSW/2$ must have a minimum regret of at  
$\MSW/2+1$, just like the first-price auction.)
\end{lemma}
{\em
\begin{proof}
The claim already holds if there is a single object and two buyers.
Assume by way of contradiction that there exists an efficient, IR, truthful
auction where the seller's revenue is at least $r\MSW$. 
Since the auction is efficient and truthful, the bidder with the higher
bid $b$ will win the auction. It follows that this bidder must pay at
least $rb$ (or else either the revenue guarantee could not be satisfied,
or the auction is not truthful), but at most $b$ (or else the auction
cannot be both truthful and IR). Thus, player 1's regret when
bidding its valuation $v$ is at least $rv-1$, since if player 2 bids
$0$, player 1 needs to pay at least $rv$, whereas he could have paid at
most $1$ by bidding $1$ (since with a truthful, IR mechanism, a player
will never pay more than he bids).
\qed
\end{proof}
}

{\em The following result shows that no regret-minimizing truthful
mechanism can do significantly better than the first-price auction 
in terms of 
maximizing revenue.} 
{\em
\begin{lemma} \label{combauc.clm2} 
No efficient, IR, regret-minimizing truthful mechanism can guarantee
the seller more than $((\sqrt{5}-1)/2)\MSW$ of revenue.
\end{lemma}
\begin{proof}
As in Lemma~\ref{combauc.clm1}, consider a mechanism for a single
object case with two buyers that has a revenue guarantee of $r\MSW$.  
We claim that if player $1$ has valuation $v$, then his regret if he
bids $\alpha v$ is at most $\max(\alpha v,v - r \alpha v )$. 
To see this, note that if
player 2 bids 
$b \le \alpha v$, then player 1 pays at most $\alpha v$ (by IR and truthfulness),
which potentially could have been saved.
Thus, his regret is at most $\alpha v$ if player 2 bids less than $\alpha v$.
If, on the other hand, player 2 bids $b > \alpha v$,
then player 1 needs to pay at least $r\alpha v$ to win the object 
(by truthfulness and the revenue guarantee), so his regret is at most
$v-r\alpha v$. 
It is easy to see that $\alpha v = v  - r\alpha v$ if $\alpha = 1/(r+1)$.
Thus, if $\alpha = 1/(r+1)$, player 1's regret is at most $v/(r+1)$.
(We are ignoring here the possibility that $v/(r+1)$ is not an integer,
hence not a legal bid.  As we shall see, this will not be a problem.)
But, by Lemma~\ref{combauc.clm1}, player 1's regret when his valuation
is $v$ can be as high as $rv-1$.  Thus, we must have $v/(r+1) \ge rv -
1$, or equivalently, $(r - 1/(r+1))v \le 1$.  This can be guaranteed for
all $v$ only if $r - 1/(r+1) < 0$, so we must have $r <(\sqrt{5} - 1)/2$.
\qed
\end{proof}
}
\end{example}

\commentout{
\section{Hierarchical Deletion}
\label{ihe.sec}

\commentout{
\begin{itemize}
\item where does the hierarchy come from.
\item what if we consider actions that survive no matter what the hierarchy is?
\item what if we consider actions that survive for \emph{some} hierarchy? i.e., consider ``assessment's'' that consists of a strategy, hierarchy pair.
\item why uniform over all the moves in the hierarchy? why not also provide some distribution (belief) over the hierarchy.
\item combining the above two: if we view the hierarchy and the belief over the hierarchy as a justification for a move, does this give a notion of rationalizability?
\item why Nash Equilibrium as Eq concept?
\item what to do if multiple Nash (if neither Pareto dominates the other)?
\end{itemize}
}
The second solution concept we consider, \emph{hierarchical
deletion}, 
takes as its point of departure  the observation that,
when dealing with a complicated decision problem, one standard approach
used in practice is to make a high-level decision, and
then worry later about lower-level details.  
Just as in Section \ref{irm.sec}, we start by showing how this idea can be applied to the
Traveler's Dilemma, and then provide a formal definition.

Consider an an agent trying to decide what to play in Traveler's Dilemma.
She needs to start considering what the other agent will play, and how
likely he is to play it.  This, in turn, requires considering what the
other agent thinks that she will play, and so on.  To simplify the
problem, the agent might decide that, as a first step, she will make a
high-level decision: should she play high or should she play low, where 
``playing high'' is interpreted as playing in the interval $[51,100]$,
and ``playing low'' is interpreted as playing in the interval $[2,50]$.
(As we shall see, there is nothing special about this particular split
into ``high'' and ``low''; the approach is extremely robust.)

Formally, we consider a variant of the Traveler's Dilemma, 
where both agents have two actions, $l$ and $h$.  Intuitively, choosing
$l$ corresponds to playing in the interval $[2,50]$ and choosing $h$
corresponds to playing in the interval $[51,100]$.  Thus, 
there are four possible outcomes: $(h,h)$, $(h,l)$, $(l,h)$,
and $(l,l)$.  What payoff should we assign to each of these outcomes?
We think of an outcome like $(h,h)$ as corresponding to
2,500 outcomes---all the outcomes of the form $(a,b)$ where $51
\le a, b \le 100$, and similarly for the other three high-level
actions. Of course, the agent does not want to think carefully about what 
probability to assign to each of these outcomes; the point of making a
``high-level'' decision is to avoid such considerations.  
To deal with this, we take the payoff of the high-level actions to be
the expected payoff over all the associated low-level actions, assuming
that each is equally likely.  A straightforward computation shows that
if $p$ is the penalty/reward, then the payoffs in this case are
given by the following table:
\begin{figure}[htb]   
\begin{center}   
\begin{tabular}{l|c|c|}\multicolumn{1}{l}{}&\multicolumn{1}{c}{$h$}&\multicolumn{1}{c}{$l$}\\      
\cline{2-3}    
 $h$& $(67.19,67.19)$& $(26-p, 26+p)$ \\ \cline{2-3}   
 $l$&$(26+p,26 - p)$&$(16.61,16.61)$\\ \cline{2-3}   
\multicolumn{1}{l}{\vspace{-3mm}}   
\end{tabular}\\   
\end{center}   
\end{figure}

It is easy to see that both $(h,h)$ and $(l,l)$ are Nash equilibria as
long as $p < 42$. Moreover, if $p < 42$, then $(h,h)$ Pareto dominates
$(l<l)$: both players do better with $(h,h)$.  

If $p$ is indeed less than $42$, we can assume that both players play
$h$, and focus on the interval $[51,100]$.  We can again divide it into
a low and 
a high subinterval, $[51,75]$ and $[76,100]$.  And again, it turns out that
$(l,l)$ and $(h,h)$ are the only Nash equilibria (under the
assumption that we assign equal probability to all outcomes compatible
with the high-level actions), and that $(h,h)$ Pareto dominates
$(l,l)$ as long as $p$ as sufficiently small ($p < 21$ will do).  

We can clearly continue in this way.  
\fullv{
To understand what happens in
general, let $\Gamma^{m_1,m_2,m_3,p}$, where $2 \le m_1 < m_2 < m_3 \le 100$
and $p \ge 2$ is the penalty/reward, be the
variant of the Traveler's Dilemma where the low interval in $[m_1,m_2]$,
the high interval is $[m_2+1,m_3]$, and $p$ is the punishment/reward
(so that we initially considered $\Gamma^{2,50,100,p}$).  A relatively
straightforward computation shows that the payoffs in this game are
as described in the following table, where $a$ and $b$ are independent
of $p$,  
$a$ is approximately (i.e., within 1 of)  $m_1 +(m_2-m_1-1)/3$, and $b$ is
approximately $(m_2+1 + (m_3 - m_2-1)/3$.  
\begin{figure}[htb]   
\begin{center}   
\begin{tabular}{l|c|c|}\multicolumn{1}{l}{}&\multicolumn{1}{c}{$h$}&\multicolumn{1}{c}{$l$}\\      
\cline{2-3}    
 $h$& $(a,a)$& $(\frac{m_1 + m_2}{2} - p, \frac{m_1+m_2}{2} + p)$ \\
 \cline{2-3}    
 $l$&$(\frac{m_1+m_2}{2} + p, \frac{m_1+m_2}{2} - p)$&$(b,b)$\\ \cline{2-3}\\ 
\multicolumn{1}{l}{\vspace{-3mm}}   
\end{tabular}\\   
\end{center}   
\end{figure}   
}
Suppose that $p=2$, as in the original story, and we keep subdividing
intervals  
in half,  where, for definiteness, if
an interval has odd length, we subdivide it into high and low so that the
high interval is the larger one.  It follows that both $(h,h)$ and
$(l,l)$ are Nash equilibria, with $(h,h)$ Pareto dominating $(l,l)$,
until we consider the interval $[97,100]$.  In this case, the high
interval is  $[99,100]$ and the low interval is $[97,98]$, and we get the 
following payoff table:
\begin{figure}[htb]   
\begin{center}   
\begin{tabular}{l|c|c|}\multicolumn{1}{l}{}&\multicolumn{1}{c}{$h$}&\multicolumn{1}{c}{$l$}\\      
\cline{2-3}    
 $h$& $(99.25, 99.25)$& $(95.5,99.5)$ \\ \cline{2-3}   
 $l$&$(99.5,95.5)$ &$(97.25,97.25)$\\ \cline{2-3}   
\multicolumn{1}{l}{\vspace{-3mm}}   
\end{tabular}\\   
\end{center}   
\end{figure}   

\noindent Now $(l,l)$ is the only Nash equilibrium.  Thus, we focus on the
interval $[97,98]$.  Again, if we partition this in the obvious way into
high and low, $(97,97)$ is the only equilibrium.  Thus, this process
leads to exactly what turned out to be the best strategy in the
empirical study of Becker, Carter, and Naeve \citeyear{BCN05}.
Interestingly, the three most often-played pure strategies in that
experiment were 100, 98, and 97 (in that order).

So far, our discussion has considered the case that $p=2$.  For larger
values of $p$, it is easy to see that $(l,l)$ becomes the only Nash
equilibrium earlier in the process, and the equilibrium reached by
iterating the process is smaller.  For example, if $p=3$, if we do a
binary partition, then we end up with $(94,94)$ rather than $(97,97)$;
if $p=6$, then we end up with $(88,88)$; if $p > 41$, then we end up
with $(2,2)$.  Thus, the dependence on $p$ here is significant.  Greater
penalty/reward results in smaller equilibria, again giving us the same
qualitative behavior as in the experiments of Capra et
al.~\citeyear{CGGH99}. 

There is nothing special about the binary split here.  \emph{Any} way of
splitting into two intervals that has the property that both intervals
have at least two elements if possible (that is, if the original
interval had at least four elements) leads to a
solution of 97 or 98.  (We can get 98 if, for example, we end up with an
interval $[95,100]$, which we subdivide into $[95,97]$ and $[98,100]$.
$(h,h)$ Pareto dominates $(l,l)$, leading us to consider the interval
$[98,100]$.  Now both possible divisions make $(l,l)$ the only Nash
equilibrium, and we end up with $(98,98)$.)  Nor is there anything
special about just splitting into two subintervals.  If instead we use
three subintervals, low, medium, and high, we again get to 97 or
98.  

Of course, it is crucial that we divide intervals into intervals.
Taking the high-level actions to be ``ask for an even amount''
and ``ask for an odd amount'' would not work.  Choosing
to divide up into intervals here seems far more natural than dividing
into even and odd; we discuss this issue further below. 
It is not essential to use averaging to calculate the
utility of the strategies $h$ and $l$.  A number of other reasonable
utility functions would lead to essentially the same outcome 
(with simpler computations!).   For example, suppose that we take the
utility of strategy profile like $(h,h)$ for player 1 is taken to be the
\emph{minimum} of the utilities to player 1 of the strategy profiles
$(s_1,s_2)$ where $s_1$ and $s_2$ are high strategies, and similarly for
player 2.  Then, for the initial split into $[51,100]$ and $[2,50]$
with $p=2$, the utility of $(h,h)$ is 49 to each player; $(h,l)$ has
utility $(0,2)$, $(l,h)$ has utility $(2,0)$, and $(l,l)$ has utility
$(2,2)$.  Again, $(h,h)$ and $(l,l)$ are both Nash equilibria, and
$(h,h)$ Pareto dominates $(l,l)$.  Continuing in this way again leads to 
$(97,97)$ as the only outcome; just as for averaging, with the partition
$[97,98]$ and $[99,100]$, $(l,l)$ is the only Nash equilibrium.  And,
again, greater penalty/reward results in smaller equilibria.  

It is quite standard in many real-life situations to focus on the
high-level issues before worrying about the details.  What this example
shows is that this kind of hierarchical reasoning is also relevant in
games.  

\subsection{Formal Definitions}
While the intuitions behind the hierarchical decomposition approach are
fairly simple, formalizing the approach is not easy.  Roughly speaking, 
a hierarchical decomposition defines a sequence of games, where the
``moves'' of the games involve choosing some set of actions in the
underlying game.  We must give each such move a utility.

We start with the definition of hierarchical deletion
for pure strategies.
Let $G=(\Play,A,\vec{u})$ be a strategic game. 
Let $H_i$ be a tree where the leaves correspond to the actions in $A_i$;
that is, there is a 1-1 mapping from the leaves of $H_i$ to $A_i$.
This tree formalizes the hierarchical decomposition of the strategies of
player $i$.    
A node in the tree can be associated with the set of leaves (actions)
below it.  Thus, we can identify a node in $H_i$ with a subset of
$A_i$; for each $k$, the nodes at depth $k$ in $H_i$ partition
$A_i$.  The first level of the tree represents the first step
of the decomposition (into the intervals $[2,50]$ and $[51,100]$ in the
case of the Traveler's Dilemma); the second level represents the
refinements of the first decomposition, and so on.
We assume that the trees $H_1, \ldots, H_n$ all have the same height.

The players then play a sequence of games, one for each level of the
hierarchy.  Intuitively, in the $k$th game, the $k$th step of the
decomposition is chosen.  Moreover, it is implicitly assumed that these
choices are made ``myopically'':  the choice made in the $k$th game is what
seems to be the optimal choice in the $k$th game, given some appropriate
notion of optimality, without trying to compute what the
choices will be in the future.   
More precisely, in the $k$th game, player $i$'s action set is the 
set of nodes $B$ such that $(A',B)$ is an edge
in $H_i$, where $A'$ is the node at depth $k$ in $H_i$ corresponding to
the move made by player $i$ in the $(k-1)$st game; that is, player $i$'s
moves in the $k$th game amount to the possible further decompositions
he can make, given his current decomposition. 
To fully specify the $k$th game we also require a utility function $\vec{u}^k$,
which is the utility that the agents are myopically trying
to optimize. 

\begin{definition} A \emph{hierarchical decomposition game} based on a
strategic game $G= (\Play,A,\vec{u})$ is a tuple of the form $(H_1,
\ldots, H_n, \vec{u}^1, \ldots, \vec{u}^m)$ for some $m$, where 
\begin{itemize}
\item $H_1, \ldots, H_n$ are trees of depth $m$; 
\item there is a 1-1 map from the leaves of $H_i$ to $A_i$;
\item $\vec{u}^k$ is a utility function on profiles of nodes at depth $k+1$
\end{itemize}
\end{definition}

We assume that the hierarchical decomposition $(H_1, \ldots, H_n)$ and
the utility functions $\vec{u}^1, \ldots, \vec{u}^m$ are defined exogenously.
However, we still need to define an appropriate solution concept for
this game.  
Following the analysis in the example, we start by considering
\emph{pure hierarchical equilibria}.  Intuitively, we choose a pure Nash
equilibrium at each step
with a preference for Pareto-optimal equilibria.
Thus, a pure hierarchical equilibrium for the game is a profile
$(\vec{B}_1, \ldots, 
\vec{B_n})$, where $\vec{B}_i = (B_i^1, \ldots, B_i^m)$ represents a
path in $H_i$, so that $B_i^1 \supseteq B_i^2 \supseteq \ldots \supseteq
B_i^m$, and $B_i^m$ is a leaf of $H_i$, so that $B_i^m$ corresponds to an
action in $G$.  Moreover, $(B_1^1, \ldots, B_n^1)$ is a Nash equilibrium
in the 
first game of the hierarchical decomposition and is Pareto optimal if
there exists a Pareto optimal Nash equilibrium; if $1 < k \le m$,
then $(B_1^k, \ldots, B_n^k)$ is a Nash equilibrium in the $k$th game of
the hierarchical decomposition conditional on being at 
$(B_1^{k-1}, \ldots, B_n^{k-1})$,%
\footnote{That is, we consider the game where the only possible moves
for the players are the ones leading from $(B_1^{k-1}, \ldots,
B_n^{k-1})$, and the players play a Nash equilibrium of this restricted game.}
and is Pareto optimal if there is a
Pareto optimal Nash equilibrium.  
The \emph{outcome} of this hierarchical deletion is the action profile
$(B_1^m, \ldots, B_n^m)$ (note that this is an action profile in the
underlying game).  

A natural condition on the utility function $\vec{u}^k$
would be to require that it ``respects'' the utility of the associated action profiles in the underlying game; more
precisely, one could require that 
\begin{equation}\label{eq:utility}
\min\{u_i(\vec{a}): \vec{a} \in B_1 \times \cdots \times B_n\} \le
u^k_i(\vec{B}) \le 
\max\{u_i(\vec{a}): \vec{a} \in B_1 \times \cdots \times B_n\}.
\end{equation}
In our example, we took $u^k_i(\vec{B})$ to be the average of player
$i$'s utilities for the action profiles in $B_1 \times \cdots \times
B_n$, treating each one as equally likely.
This is easily seen to satisfy (\ref{eq:utility}).  
Utility functions such as  $\min$ or $\max$ also satisfy (\ref{eq:utility}).
While (\ref{eq:utility}) seems reasonable, especially when the set of
actions in the cell of the decomposition is small, it may not always be
appropriate in cells with large numbers of actions.  Indeed, computing
the exact value of the average utility for the the Traveler's Dilemma
with large sets of action profiles involves a somewhat long and tedious
computation.  
An agent might not be prepared to go through the exact computation, and
instead take some approximation (or use functions like min or max that
are easier to compute).

Roshambo provides a simple example of a game with no pure Nash
equilibria. 
For similar reasons, it is easy to show that a pure hierarchical
equilibrium may not exist in a hierarchical game.  We thus now define a
notion of \emph{mixed} hierarchical equilibrium.  Not surprisingly, the
idea is that, at each stage, we choose a distribution over moves, rather
than a single move.  The key insight is that we can view the $k$th game
in a hierarchical game as a 
Bayesian game.
The possible types of player $i$
in the $k$th game are the nodes at depth $k$ in $H_i$.  Thus, player
$i$'s types in the $k$th game are disjoint subsets of $A_i$.  In a
Bayesian game there is also assumed to be a (commonly known)
distribution over type profiles.
A mixed strategy for player $i$ in a Bayesian game is a function from 
$i$'s types to a distribution over actions.  
A \emph{mixed hierarchical equilibrium} for the game $G$ 
is a profile
$(\vec{\mu}_1, \ldots, \vec{\mu}_n)$, where $\vec{\mu}_i = (\mu_i^1,
\ldots, \mu_i^m)$, and $\mu_i^j$ is a mixed strategy for player $i$ in
the $j$th game.  Moreover, $(\mu_1^1, \ldots, \mu_n^1)$ is a (mixed
strategy) Nash equilibrium in the 
first game of the hierarchical decomposition and is Pareto optimal if
there exists a Pareto optimal Nash equilibrium; if $1 < k \le m$,
then $(\mu_1^k, \ldots, \mu_n^k)$ is a Nash equilibrium in the $k$th
game of the hierarchical decomposition 
where the common prior over types is taken to be that generated by 
$((\mu_1^1, \ldots, \mu_1^{k-1}), \ldots, (\mu_n^1, \ldots,
\mu_n^{k-1}))$ (i.e., by what has happened in the earlier games). 
Note that a mixed hierarchical equilibrium determines a mixed strategy
in the underlying game. 
Also note that that---since every finite Bayesian
game has a Nash equilibrium---every finite
strategic game has a mixed hierarchical equilibrium for every hierarchical decomposition.
Finally, we can define hierarchical deletion and hierarchical
equilibrium in Bayesian games.  The definition is essentially the same
as that of mixed hierarchical equilibrium.  
Just as in a strategic game, a mixed hierchical equilibrium in a
Bayesian game is a profile  
$(\vec{\mu}_1, \ldots, \vec{\mu}_n)$, where $\vec{\mu}_i = (\mu_i^1,
\ldots, \mu_i^m)$, and $\mu_i^j$ is a mixed strategy for player $i$ in
the $j$th game.  Each game is a Bayesian game, so the strategy for
player $i$ can depend on player $i$'s type.  Just as in a mixed
hierarchical equilibrium, $(\mu_1^1, \ldots, \mu_n^1)$ is a Nash
equilibrium of the first game, 
$((\mu_1^1, \ldots, \mu_1^{k-1}), \ldots, (\mu_n^1, \ldots,
\mu_n^{k-1}))$ determines a common prior on the $k$th game, and 
$(\mu_1^k, \ldots, \mu_n^k)$ is a Nash equilibrium of the $k$ game.
We leave details to the reader.

\subsection{Comparison To Other Solution Concepts}

Observe that hierarchical decomposition 
can be viewed as 
a generalization of Nash equilibrium;  given a game, 
$G = ([n],A,\vec{u})$, consider the trivial 
hierarchy consisting of a tree of depth 1, and taking
$\vec{u}^1 = \vec{u}$. 
A mixed hierarchical equilibrium is then just a Pareto-optimal Nash
equilibrium.  
We also note that by appropriately setting the utility functions $\vec{u}^1, .., \vec{u}^k$ we can
Of course, hierarchical deletion is more interesting in cases where
it is not just a Nash equilibrium.  
As we observed at the beginning of this section,
in the Traveler's Dilemma, using a ``binary'' decomposition and the
``averaging'' utility function, 
we get 
that $(97,97)$ is the only pure hierarchical equilibrium.
Essentially the same arguments show that this is also the only mixed
hierarchical equilibrium.  
\begin{example} Centipede Game: \emph{
We show that also for the Centipede Game, hierarchical deletion
yields outcomes with the same qualitative behavior as outcomes obtained
by iterated regret minimization. 
Let $G = ([n],A,\vec{u})$ be the Centipede Game with linear utilities
when $p =2$, $k=10$. Consider the hierarchy $H_1=\{\{[1],[3],[5]\},
\{[7],[9]\}\}$, $H_2=\{\{[2],[4],[6]\}, \{[8],[10]\}\}$, 
and as in the Traveler's Dilemma let $u^k_i(\vec{B})$ be the average of player
$i$'s utilities for the action profiles in $B_1 \times B_2$, treating each one as equally likely.
For the first step of the decomposition, note that $(\{[7],[9]\},
\{[8],[10]\})$ is a Pareto-optimal Nash equilibrium. 
In step two, we are left with the game described in Figure 1.
\begin{figure}[htb]   
\begin{center}   
\begin{tabular}{l|c|c|}\multicolumn{1}{l}{}&\multicolumn{1}{c}{$[7]$}&\multicolumn{1}{c}{$[9]$}\\      
\cline{2-3}    
 $[8]$& $(7, 5)$& $(6,8)$ \\ \cline{2-3}   
 $[10]$&$(7,5)$ &$(9,7)$\\ \cline{2-3}   
\multicolumn{1}{l}{\vspace{-3mm}}   
\end{tabular}\\   
\end{center}   
\caption{The second game in the hierarchy, with linear payoffs.}
\label{fig1}
\end{figure}   
}

\noindent\emph{Here, the only Nash equilibrium is $([7],[8])$, which thus is
the only strategy that survives hierarchical deletion. 
(Recall that, in contrast, $([1],[2])$ is the only strategy that
survives iterated deletion of weakly dominated strategies,
is the only one that is rationalizable, and is the only Nash
equilibrium.) 
}

\emph{
Now consider the Centipede Game with exponential payoffs, using
the same hierarchy and defining $u^k$ analogously. It follows using the
same argument  
that in the first step of the decomposition, $(\{[7],[9]\},
\{[8],[10]\})$ is a Pareto optimal Nash equilibrium. In step two, we are
instead left with the game described in Figure~\ref{fig2}.
}

\begin{figure}[htb]   
\begin{center}
\begin{tabular}{l|c|c|}\multicolumn{1}{l}{}&\multicolumn{1}{c}{$[7]$}&\multicolumn{1}{c}{$[9]$}\\      
\cline{2-3}
 $[8]$& $(2^7+1, 2^7-1)$& $(2^8-1,2^8+1)$ \\ \cline{2-3}   
 $[10]$&$(2^7+1, 2^7-1)$ &$(2^9-1,2^9+1)$\\ \cline{2-3}   
\multicolumn{1}{l}{\vspace{-3mm}}   
\end{tabular}\\   
\end{center}   
\caption{The second game in the hierarchy, with exponential payoffs.}
\label{fig2}
\end{figure}    

\noindent\emph{Here, the only Nash equilibrium is $([9],[10])$, 
so now this is the only strategy profile that survives
hierarchical deletion.
}
\qed \end{example} 

\commentout{
$H = H_1 \times \ldots \times H_n$; more precisely $u'_i$ is 
player $i$'s utility function in a generalized extensive-form game where
player $i$ moves at each node in the decomposition tree $H_i$, and at a
node in the tree, much choose a successor.  Thus, for example, at the
top of the decomposition tree for the Traveler's Dilemma, each player
must choose whether to play $h$ or $l$, and has a similar choice at each
subsequent node.
Additionally we require
that $u'(a_i, \vec{a}_{-i}) = u(a_i, \vec{a}_{-i})$ for all action
profiles $\vec{a} 
in A$, i.e., that $u'$ is consistent with $u$ when considering the
actions of the original game $G$. 
The pair $\Hcal = (H,u')$ is called a \emph{hierarchy} for the game $G$.

Let $\solve(\cdot,\cdot)$ be a function that on input a player $i'$ and a
game $G$ outputs a strategy for player $i$; $\solve(\cdot,\cdot)$ should
be  thought of as the solution concept employed to hierarchically solve
the game.  Note that we here only consider solution concepts that
uniquely determine the strategy for each player.\footnote{This can be
trivially generalized to more general solution concepts by letting the
process of 
iterated hierarchical deletion ``fail'' whenever it reached a point
where the solution is undetermined.} 

For each $i$, let $H_i^0(G) = H_i$.
Define $H_i^j(G, \Hcal)$ as follows.
If $H_i^{j-1}(G, \Hcal)$ only contains a single strategy, let
$H_i^{j}(G, \Hcal) = H_i^{j-1}(G, \Hcal)$. Otherwise, consider the game
$G'=(\Play, A',\vec{u'})$ where $A'=A'_1 \times \ldots \times A'_n$ and
$A'_{i'}$ is the set of children of the root of
$H_{i'}^{j-1}(G,\Hcal)$. Let $H_i^{j}(G, \Hcal)= \solve(i, G')$. Define
$H_{i}^{\infty} = \cap_j \Spure_i^j(G, \Hcal)$ 

\Rnote{Remark that the generalized utility function $u'$ doesn't have to
be consistent with $u$ except for the leafs. one potential way to define
is it by indicuing some belief of what will be played at this level; but
the more general variant alows us to capture situation where my belief
about the utility is  
``unjustified''; this seems to more appropriately model bounded-rational players.}

\begin{definition}
Let $G= (\Play,A,\vec{u})$ be a normal form game, and $\Hcal$ a
hierarchy for $G$. We say that an pure strategy (action) $a_i \in A_i$  
\emph{survives hierarchical deletion with respect to the
hierarchy $\cal H$} if $a_i \in H_i^{\infty}(G)$. 
\end{definition}

We generalize the notion of hierarchical deletion to consider also mixed
strategies by simply letting the leafs in the (now infinite) tree $H_i$
be all the mixed strategies for player $i$. 
}
}

\section{Related Work}\label{sec:related}
While the notion of regret has been well studied in the context of
decision theory (see the \cite{Hayashi08} and the references therein
for some discussion of the recent work).  To the best of our knowledge,
there was no work on applying regret to game theory up until very recently.
In the computer science literature, 
Hyafil and Boutilier \cite{BH04} consider
\emph{pre-Bayesian} games, where each agent has a type and a player's 
utility depends on both the action profile and the type profile, just as
in a Bayesian game, but now there is no probability on types.%
\footnote{Hyafil and Boutilier actually consider a slightly less general
setting, where the utility for player $i$ depends only on player $i$'s
type, not the whole type profile.  Modifying their definitions to deal
with the more general setting is straightforward.}
The solution concept they use is a hybrid of Nash equilibrium and
regret.  
Roughly speaking, they take regret with respect to the types of
the other players, but then use Nash equilibrium with respect to the
strategies of other players.
That is, they define $\regret^{\Scal_i}_{i}(a_i \mid \vec{a}_{-i},
\vec{t})$ as we do (taking $\Scal_i$ to consist of all strategies for
player $i$), but then define $\regret^{\Scal_i}(a_i \mid \vec{a}_{-i})$
by minimizing over all $\vec{t}_{-i}$.  They then define a profile
$\vec{\sigma}$ to be a \emph{minimax-regret equilibrium} if, for all
type profiles $\vec{t}$, no agent can decrease his regret by changing
his action.  For strategic games, where there are no
types (i.e., $|T|=1$),
their solution concept collapses to Nash  equilibrium. 
Thus, their definitions differ from ours in that they take regret with
respect to types, not with respect to the strategies of other players
as we do, and they do not iterate the regret operation.

Aghassi and Bertsimas \citeyear{AB06} also consider 
pre-Bayesian games, 
and use a solution concept in the spirit of that of Hyafil and
Boutilier.  However, rather than using minimax regret,
they use \emph{maximin}, where a maximin action
is one with the best worst-case payoff, taken over all the types of the
other agents. 
Just as with the Hyafil-Boutilier notion, the Aghassi-Bertsimas notion
collapses to Nash equilibrium if there is a single type.

By way of contrast, we assume that players minimize regret also with
respect to the strategies of all other player. Additionally, we 
iterate the deletion process.

Even closer to our work is a recent paper by Renou and Schlag
\citeyear{RS08}.  Just as we do, they focus on strategic games.
Their motivation for considering regret, and the way they do it in the
case of pure strategies, is identical to ours (although they do not
iterate the deletion process).  They
allow prior beliefs, as in Section~\ref{sec:prior}, and require
that these beliefs are described by a closed, convex set of
strategies.  They are particularly interested in strategy profiles
$\vec{\sigma}$  that minimize regret for each agent with respect to all the
strategy profiles in an $\epsilon$ neighborhood of $\vec{\sigma}$.

Note that, although they define regret for pure strategies, this
is only a tool for dealing with mixed strategies; they do not consider
the regret of a pure strategy with respect to a set of pure strategies,
as we do, because a non-singleton set of pure strategies is not convex. 
In particular, they have no analogue to our analysis of pure
strategies in Section~\ref{sec:pure}.  
If we consider regret relative to the set of all mixed strategy profiles,
then we are just in the setting of Section~\ref{sec:mixed}.  However,
their definition of regret for mixed strategies is different from ours.
Our definition of  $\regret^{\Scal_i}_{i}(\sigma_i \mid
\vec{\sigma}_{-i})$ does not depend on whether 
$\vec{\sigma}$ consist of pure strategies or
mixed strategies (except that expected utility must be used in the case of
mixed strategies, rather than utility).  By way of contrast, Renou and Schlag
\citeyear{RS08} define 
$$\regret'_{i}(\sigma_i \mid \vec{\sigma}_{-i}) = 
\sum_{a_i \in A_i, a_{-i} \in A_{-i}}
 \sigma_i(a)\vec{\sigma}_{-i}(\vec{a}_{-i})\regret_i^{A_i}(a_i \mid
 \vec{a}_{-i}) ,$$ 
where, as before, $\regret_i^{A_i}(a_i \mid \vec{a}_{-i})$ denotes the
regret of player $i$ relative to the actions $A_i$. 
That is, $\regret'_i(\sigma_i \mid \vec{\sigma}_{-i})$ is calculated
much like the expected utility of $\vec{\sigma}$ to agent $i$, in terms
of the appropriate convex combination regrets for pure strategies.  Note
that $\regret'_i$ is independent of any set $\Scal_i$.  

In general, $\regret'_{i}(\sigma_i \mid \vec{\sigma}_{-i})$ is quite
different from $\regret_i^{\Smixed_i}(\sigma_i \mid \vec{\sigma}_{-i})$,
as the following example shows.

\begin{example}\label{xam:differ}
\emph{Consider the symmetric 2-player game where 
$A_1 = A_2 = \{a,b,c\}$, and player 1's payoffs are given
by the following table:}

\begin{table}[h]
\begin{center}
\begin{tabular}{c |  c c c }
& $a$ & $b$ & $c$\\
\hline
$a$ &5 &2 &1\\
$b$  &0  &3  &1\\
$c$  &3  &1  &4  
\end{tabular}
\end{center}
\end{table}

\emph{It immediately follows that $\regret_1^{\Sigma_1}(a \mid a) =
\regret_1^{\Sigma_1}(b \mid b) = \regret_1^{\Sigma_1}(c \mid c) = 0$;
$\regret_1^{\Sigma_1}(a \mid b) = 0$; $\regret_1^{\Sigma_1} (b \mid a) = 3$;
$\regret_1^{\Sigma_1}(c \mid a) =  \regret_1^{\Sigma_1}(c \mid b) = 2$;
and $\regret_1^{\Sigma_1}(a \mid c) =  \regret_1^{\Sigma_1}(b \mid b) =3$.
If $\sigma = (1/6)a + (1/2)b + (1/3) c$, it is easy
to see that $u_1(a,\sigma) = 5/6 + 1 + 1/3 = 13/6$;
$u_1(b,\sigma) = 3/2$;  and
$u_1(c,\sigma) = 1/2 + 1/2 +  4/3 = 7/3$.  Thus, $c$ minimizes regret
relative to $\sigma$, and
$\regret_1^{\Sigma_1}(c \mid \sigma ) = 0$, 
$\regret_1^{\Sigma_1}(a \mid \sigma)  = 1/6$, and 
$\regret_1^{\Sigma_1}(b \mid \sigma) = 5/6$.  Thus,
$\regret_1^{\Sigma_1}(c \mid \sigma ) < \regret_1^{\Sigma_1}(a \mid \sigma) <
\regret_1^{\Sigma_1}(b \mid \sigma)$. }

\emph{On the other hand, $\regret'_1(a \mid \sigma) = 3/2 + 2/3 = 5/3$;
$\regret'_1(b \mid \sigma) = 1/2 + 2/3 = 7/6$; 
and $\regret'_1(c \mid \sigma) = 4/3$.
Thus, $\regret'_1(b \mid \sigma) < \regret'_1(c \mid \sigma) <
\regret'_1(a \mid \sigma)$.}
\qed \end{example}

The reason for the difference between $\regret$ and $\regret'$ in this
example is that the best responses to $a$, $b$, and $c$ are all different.
To understand the difference between the two approaches, consider
an agent who is playing $c$, facing a 
population of people, one-sixth of whom play $a$, half of whom play $b$,
and the remainder play $c$ (so that the population is emulating strategy
$\sigma$).  
What is the agent's regret if he plays $c$ against such a population?
If he only plays once, and his opponent plays $a$, should he feel regret
$2$ (if 
he had only played $a$, he would have done $2$ better against that
opponent), or should he take into account, when he considers how he would
have done, that he is playing against a randomly chosen opponent from the
population, not necessarily that particular opponent (in which case 
his regret is 0).  Using $\regret'$ corresponds to the first
approach, while $\regret^{\Smixed_i}$ corresponds to the second.  

As the following result shows, if we are considering the strategy
that minimizes regret with respect to all strategies, then it does not
matter which approach we take. 

\begin{proposition}
\label{prop:agree}
Let $G=(\Play,A,\vec{u})$ be a strategic game and 
let $\sigma_i$ be a mixed strategy for player $i$.  Then
$\regret_{i}^{\Sigma}(\sigma_i) = 
\max_{\sigma_{-i} \in \Sigma_{-i}} \regret'_i(\sigma_i \mid \sigma_{-i})$.
\end{proposition}

\begin{proof}  By Proposition~\ref{onlypure.prop},
$\regret_{i}^{\Sigma)}(\sigma_i) = \max_{\vec{a}_{-i} \in A_{-i}}
\regret^{\Scal_{i}}(\sigma_i \mid \vec{a}_{-i})$.  It is also easy
follows from the definition of $\regret'_i$ that
$\max_{\sigma_{-i} \in \Sigma_{-i}} \regret'_i(\sigma_i \mid
\sigma_{-i}) =  \max_{\vec{a}_{-i} \in A_{-i}} \regret'_i(\sigma_i \mid
\vec{a}_{-i})$.  
Thus, it suffices to show that 
$\regret^{\Scal_{i}}(\sigma_i \mid \vec{a}_{-i}) = 
\regret'_i(\sigma_i \mid \vec{a}_{-i})$ for all $\vec{a}_{-i} \in
A_{-i}$.  This is straightforward.  For suppose that $a' \in A_i$ is the
best response to $\vec{a}_i$; that is, using our earlier notation, $a' =
u_i^{A_i} (\vec{a}_{-i})$.  Then, by definition, 
$$\begin{array}{ll}
&\regret^{\Scal_{i}}(\sigma_i \mid \vec{a}_{-i})\\
= &U_i(a',\vec{a}_{-i}) - U_i(\sigma_i,\vec{a}_{-i})\\
= &U_i(a',\vec{a}_{-i}) = \sum_{b \in A_i} \sigma_i(b) U_i(b,\vec{a}_{-i})\\
= &\sum_{b \in A_i} \sigma_i(b)(U_i(a',\vec{a}_i) - U_i(b,\vec{a}_{-i})\\
= &\sum_{b \in A_i} \sigma_i(b) \regret^{A_i}(b,\vec{a}_{-i})\\
= &\regret'(\sigma_i,\vec{a}_{-i})
\end{array}
$$
\qed \end{proof}

Proposition~\ref{prop:agree} shows that, for many of the examples
in Section~\ref{sec:mixed}, it would not matter whether we had used 
$\regret'$ instead of $\regret$.  On the other hand, the difference
between the two approaches becomes more significant if there is prior
knowledge (or we do iterated regret minimization).  
While Proposition~\ref{prop:agree} shows that it does not matter how we
define regret for mixed strategies if we have no prior knowledge, the
differences between the definitions become significant if there is some
prior knowledge (which is precisely the case focused on Renou and Schlag).
It is worth nothing that when Renou and Schlag consider regret
minimization with respect to some set $\Pi = \Pi_1 \times \cdots \times \Pi_n$
of strategy profiles, their definition is essentially our notion of
generalized regret minimization with respect to 
$(\Sigma_1 \times \Pi_{-1}, \ldots, \Sigma_n \times \Pi_{-n})$; that is, each
agent $i$ puts no 
restriction on his own strategies.  For example,
if agent 1 does regret minimization in the game of
Example~\ref{xam:differ} with respect to the set $\Sigma_1 \times
\{\sigma\}$, using our definition, he would play $c$, while using
$\regret'$, as suggested by Renou and Schlag, he would play $b$.

\section{Discussion}
\label{discussion.sec}
The need to find solution concepts that reflect more accurately how
people actually play games has long been recognized.  This is a
particularly important issue because the gap between ``descriptive'' and 
``normative'' is particularly small in game theory.  An action is
normatively the ``right'' thing to do only if it is the right thing to
do with respect to how others actually play the game; thus, a good
descriptive theory is an essential element of a good normative theory.

There are many examples in the literature of games where Nash
equilibrium and its refinements do not describe what people do.
We have introduced a new solution concept, iterated regret minimization,
that, at least in some 
games, seem to capture better what people are doing than more standard
solution concepts.  

The outcomes of games like the Traveler's Dilemma and the Centipede Game
have sometimes been explained by assuming that a certain fraction of
agents will be ``altruistic'', and play the helpful action (e.g.,
playing 100 in Traveler's Dilemma or cooperating in the Centipede Game)
(cf., \cite{CGGH99}). 
There seems to be some empirical truth to this assumption; for example,
10 of 45 game theorists that submitted pure strategies in the
experiments of Becker, Carter, and Naeve \citeyear{BCN05} submitted 100.
With an assumption of altruism, then the strategies of many of the
remaining players can be explained as best responses to their
(essentially accurate) beliefs.  

Altruism may indeed be part of an accurate descriptive theory, but to
use it, we first need to decide what the ``right'' action is, and also
with what percentage agents are altruistic.  We also need to explain why
this percentage may depend on the degree of punishment (as it did in the
experiments of Capra et al.~\citeyear{CGGH99}, for example).
Iterated regret minimization provides a 
different descriptive explanation, and has some normative import as well.
In particular, it seems the most appealing when considering inexperienced, but intelligent, players that play a game for the first time. In such a setting, 
it seems unreasonable to assume that players know what strategies the other players are using (which is assumed by the Nash equilibrium solution concept).

While we have
illustrated some of the properties of iterated regret minimization, we
view this  
paper as more of a ``proof of concept''. There are clearly 
many issues we have left open.  We mention a few of 
the issues we are currently exploring here.

\fullv{\begin{itemize}}
\shortv{\begin{newitemize}}

\item As we observed in Section~\ref{sec:prior}, some behavior is well
explained by assuming that agents start the regret minimization
procedure with a subset of the set of all strategy profiles, which can be
thought of as representing the strategy profiles that the agent is
considering.  But we need better motivation for where this set is coming
from.

\item We have considered ``greedy'' deletion, where all strategies that
do not minimize regret are deleted at each step.  We could instead
delete only a subset of such strategies at each step of deletion.  It is
well known that if we do this with iterated deletion of weakly dominated
strategies, the final set is strongly dependent on the order of
deletion.  The same is true for regret minimization. Getting an
understanding of how robust the deletion process is would be of interest.

\item We have focused on normal-form games and Bayesian games.
It would also be interesting to extend regret minimization to
extensive-form games.  A host of new issues arise here, particularly
because, as is well known, regret minimization is not time consistent
(see \cite{Hayashi08a} for some discussion of the relevant issues).

\item A natural next step would be to apply our solution concepts 
to mechanism design beyond just auctions.
\fullv{\end{itemize}}
\shortv{\end{newitemize}}

\appendix

\section{An Epistemic Characterization Using Kripke Structures}
\label{sec:kripke}
Let $G = ([n],A, \vec{u})$ be a strategic game and let $\Scal$ denote
the full set of mixed strategies. 
We consider a Kripke structure $(W,R_1, \ldots, R_n)$ for $(G,\Scal)$,
where $W$ denotes the set of possible worlds and $R_i$ is a
binary relation on $W$ partitioning $W$ into cells. Intuitively,
$(w_1,w_2) \in R_i$ means that player $i$ cannot distinguish the worlds
$w_1$ and $w_2$.  Let $P_i(w)$ denote player $i$'s cell in the world
$w$; that is, $P_i(w) = \{w' | (w',w) \in R_i\}$. 

At each world $w \in W$, there is an associated strategy profile
$\vec{\sigma} \in \Scal$  
and, for each player $i$, a sequence $\langle B_0^i, B_1^i, ... \rangle$ of 
subsets of player $i$'s cell at $w$.  Intuitively, $B_0^i$ are the worlds
in the cell that player $i$ considers most likely, $B_1^i$ are less
likely, and so on.  Thus, the sequence models player $i$'s beliefs.   
We assume that at each world in a cell for player $i$, player $i$ uses
the same strategy and has the same beliefs.  Thus, 
a player knows his own strategy and beliefs; his beliefs are only
over the other players' strategies and beliefs.  
Let $\strategy(w)$ denote the strategy profile
associated with $w$.  Given a set $B$ of worlds, let $\strategy(B)$ denote
the set of strategy profiles associated with the worlds in $B$.   

Note that while the notion of a lexicographic belief sequence defined
in  
Section \ref{sec:characterization} considers a sequence $(\Scal_0,
\Scal_1, \ldots)$ of strategies,  
a belief sequence here considers a sequence $\B = \langle B_0, B_1,
... \rangle$ of worlds. 
Given such a belief 
sequence
$\B = \langle B_0, B_1, ... \rangle$, let $\strategy(\B) = 
\langle \strategy(B_0), \strategy(B_1), ... \rangle$ denote the lexicographic belief sequence associated with $\B$.
Player $i$ is \emph{rational in world $w$} if player $i$'s
strategy in 
$w$ is rational with respect to the lexicographic belief sequence
$\strategy(\B^i)$, 
where $\B^i$ denotes player $i$'s belief sequence in $w$ (and
rationality with respect to lexicographic sequences is defined as in
Section \ref{sec:characterization}). 

We will be interested in Kripke structures $(W,R_1, \ldots R_n)$ that
are \emph{complete}; this is a richness assumption that is analogous to
one made by  
Brandenburger et al.~\citeyear{BFK04}.
The structure $(W,R_1, \ldots R_n)$ for $(G,\Scal)$ is 
\emph{complete} if
\begin{enumerate}
\item for every $\sigma \in \Scal$, there exists some world $w$ such that
$\strategy(w)=\sigma$; and
\item for every every player $i$ and sequence $\<\Scal^0_{-i},
\Scal^1_{-i}, \ldots\>$ of sets of
strategy profiles such that, for all $j$, $\Scal^j_{-i} \subseteq
\strategy(P_i(w))$,   
there exists a possible world $w' \in P_i(w)$ 
such that $i$'s belief sequence at $w'$ is $\B^i$, and 
$\strategy(\B^i) = \langle \Scal^0_{-i}, \Scal^1_{-i}, \ldots
\rangle$. 
\end{enumerate}
Let $(W,R_1, \ldots R_n)$ be a complete Kripke
structure for $(G,\Scal)$. 
Define a sequence of worlds $(W_0, W_2, \ldots)$, where $W^k$
intuitively consists of  
all worlds where players are rational and have rational beliefs up until
level $k-1$.
Thus, $W^k$ represents worlds where all player use at least $k$ levels
of rationality. 
More formally, consider the following sequence.
\begin{itemize}
\item $W_0$ is the subset of $W$ where each player $i$ is rational.
\item $W_1$ is the subset of $W_0$ where each player $i$'s level-0
belief $B_0^i$ 
    is such that $\strategy(B_0^i) = \Scal_{-i}$; that is, each players
    considers 
    all strategy profiles possible at the top level.
(This captures the intuition that players' primary beliefs are such that
    they make no assumptions about the other players' strategies.) 
Let $\Scal^1$ be the set of
    strategy profiles that appear in worlds in $W_1$.
\item $\ldots$
\item $W_k$ is the subset of $W_{k-1}$ where each player i's
    level-$(k-1)$ belief 
    $B_{k-1}^i$ is such that $\strategy(B_{k-1}^i) = \Scal^{k-1}_{-i}$;
    that is, 
each player's level-$(k-1)$ belief is that all players use strategies
    and beliefs from a world in $W_{k-1}$.  
(This captures the intuition that players' level-$(k-1)$ belief is that
the other players use at least $k-1$ levels of rationality.)
Let $\Scal^k$ be the set of strategy profiles that appear
    in worlds in $W_k$.
\end{itemize}

It now easily follows by induction that $\Drm^{k}(\Scal) = \strategy(W_k)$.
Completeness is required to ensure that for every world $w \in W_k$ and
every player $i$, there exists 
some world $w' \in W_k$ such that $\strategy(w) = \strategy(w')$,  for
every player $i$, 
$i$'s $k$th-order beliefs in $w$ 
and $w'$ are the same (i.e., the first $k-1$-level beliefs are the same in $w$ and $w'$),
and for every $j \geq k-1$, 
$i$'s level-$j$ belief in $w'$ is the same as $i$'s level-$(k-1)$
belief in $w'$; this ensures that $\Scal^k = \Drm(\Scal^{k-1})$.

We conclude that $\Drm^{\infty}(\Scal) = \cap_{k \in N} W_k$.
Intuitively, this means that a strategy that survives iterated regret
minimization is  
used in a world where each player is rational, each player's primary
belief is that everyone else  
is using an arbitrary strategy, each player's secondary belief is that
everyone else is rational,  
and so on.

\section{Proofs}
We provide proofs of all the results not proved in the main text here.
We repeat the statements for the convenience of the reader.

\medskip

\othm{nonempypure.prop}
Let $G=(\Play,A,\vec{u})$ be a strategic game. 
If $\Scal$ is a closed, nonempty set of strategies of the form
$\Scal_1 \times \ldots \times \Scal_n$, then 
$\Drm^\infty(\Scal)$ is nonempty, $\Drm^\infty(\Scal) = \Drm^\infty_1(\Scal)
\times \ldots \times \Drm^\infty_n(\Scal)$, and 
$\Drm(\Drm^\infty(\Scal))
=\Drm^\infty(\Scal)$.
\eothm

\begin{proof}
We start with the case of pure strategies, since it is so simple.
Since $\D(\Scal) \subseteq \Scal$ for
any deletion operator $\D$ and set $\Scal$ of strategy profiles, 
when
we have
have equality, then clearly $\Drm^{\infty}(A) =
\Drm^{k}(A)$.  Since $A$ is finite by assumption,
after some point we must have equality.
Moreover, we have $\Drm(\Drm^\infty(A)) = \Drm^\infty(A)$.
To deal with the general case, we must work a little harder.  
The fact that $\Drm^\infty(\Scal) = \Drm_1^\infty(\Scal) \times \ldots
\times \Drm_n^\infty(\Scal)$ is straightforward and left to the reader.
To prove the other parts, we first need the following lemma.

\begin{lemma}
\label{closed.lem}
Let $\Scal$ be a nonempty closed set of strategies. Then $\Drm(\Scal)$ is 
closed and nonempty.
\end{lemma}
\begin{proof}
We start by showing that $regret_i^{\Scal_i}$ is continuous.
First note that $\regret^{\Scal_i}_{i}(a_i \mid
\vec{a}_{-i})$ is a continuous function 
of $\vec{a}_{i}$.  By the closedness (and hence compactness) of $\Scal$ it 
follows that $\regret_{i}^{\Scal}(a_i) = 
\max_{\vec{a}_{-i} \in \Scal_{-i}}
\regret^{\Scal_{i}}(a_i \mid \vec{a}_{-i})$ is 
well defined (even though it involves
a max).  To see that $\regret_i^{\Scal}$ is continuous, suppose not.  This 
means
that there exist some $a$, $\delta$ such that for all $n$, there exists an
$a_n$ within $1/n$ of $a$ such that 
$|\regret_i^{\Scal}(a_n) - \regret_i^{\Scal}(a)| > \delta$.  By compactness of $\Scal$, it follows by the Bolzano-Weierstrass theorem \cite{Rudin76}
that there exist a convergent subsequence $(a_{n_m},\regret_i(a_{n_m}))$ which converges to $(a,b)$.  
We have $|b - \regret_i(a)| \geq \delta$, which is a contradiction.  
Now, to see that $\Drm(\Scal)$ is nonempty, it suffices to observe that,
because $\Scal$ is compact, and $\regret_i^{\Scal}$ is continuous,
for each player $i$, there must be some
strategy $\sigma_i$ such that $\regret_i^{\Scal}(\sigma_i) =
\minregret_i^\Scal$.  Thus, $\sigma_i \in \Drm(\Scal)$.

To show that $\Drm(\Scal)$ is closed, suppose that
$\<\sigma^m\>_{m= 1, 2, 3, \ldots}$ is a sequence of
mixed strategy 
profiles in $\Drm(\Scal)$ converging to $\sigma$ (in the sense that the
probability placed by $\sigma^m$ on a pure strategy profile converges to
the probability placed by $\sigma$ on that strategy profile) 
and, by way of contradiction, that $\sigma \notin \Drm(\Scal)$.  
Thus, for some player $i$, $\sigma_i \notin \Drm_i(\Scal)$.  
Note that, since $\Drm(\Scal) \subseteq \Scal$, the sequence
$\<\sigma^m_i\>_{m= 1, 2, 3, \ldots}$ is in $\Scal$; since 
$\Scal$ is closed, $\sigma \in \Scal$.  
Let $\minregret_i^{\Scal} = b$.  Since $\sigma_i^m \in \Drm_i(\Scal)$, 
we must have that $\regret_i^{\Scal}(\sigma_i^m) = b$ for all $m$.
Since $\sigma_i \notin \Drm_i(\Scal)$, 
there must exist some strategy profile $\vec{\tau }\in \Scal$
such that $U_i(\vec{\tau}) - U_i(\sigma_i,\tau_{-i}) = b' > b$.  
But by the continuity of utility, $\lim_{m} U_i(\sigma_i^m,\tau) = b'$.
This contradicts the assumption that 
$\regret_i^{\Scal}(\sigma_i^m) = b$ for all $m$.
Thus, $\Drm(\Scal)$ must be closed. \qed
\end{proof}

Returning to the proof of Proposition~\ref{nonempypure.prop}, note that 
since $\Scal$ is closed and nonempty, it
follows by Lemma \ref{closed.lem}  
that $\Drm^k(\Scal)$ is a closed nonempty set for all $k$.
Additionally, note that $\Drm^k(\Scal)$ can be viewed as a subset of
the compact set $[0,1]^{|A|}$ (since a probability distribution on
a finite set $X$ can be identified with a tuple of numbers in
$[0,1]^{|X|}$); it follows that $\Drm^k(\Scal)$ is also
bounded, and thus compact.  
Finally, note that the set $\{\Drm^k(\Scal): k = 1, 2,
3, \ldots\}$ has the \emph{finite intersection property}: the intersection
of any finite collection of its elements is nonempty (since it is equal to 
the smallest element). The compactness 
of $\Scal$ now guarantees that the intersection of all the sets in a
collection of closed subsets of $\Scal$ with the finite
intersection property is nonempty \cite{Munkres}.
In particular, it follows that $\Drm^\infty(\Scal)$ is
nonempty.

To see that $\Drm^\infty(\Scal)$ is a fixed point of the deletion
process, suppose, by way of contradiction, that $\sigma_i \in 
\Drm_i^\infty(\Scal) - \Drm_i(\Drm^\infty(\Scal))$.
Let $\minregret_i^{\Drm^\infty(\Scal)} = b$ and choose $\sigma'_i \in
\Drm^{\infty}_i(\Scal)$ such that
$\regret_i^{\Drm^\infty(\Scal)}(\sigma_i') = b$.  
Since $\sigma_i \notin \Drm(\Drm^\infty(\Scal))$, it must be the
case that $\regret_i^{\Drm^\infty(\Scal)}(\sigma_i') = b' > b$.  
By assumption, $\sigma_i
\in \Drm^\infty(\Scal)$, so $\sigma_i \in
\Drm^k(\Scal)$ for all $k$;  
moreover, $\regret_i^{\Drm^k(\Scal)}(\sigma_i) \ge b'$.  
Since $\sigma_i \in \Drm^{k+1}(\Scal)$, it follows that
$\minregret_i^{\Drm^k(\Scal)} \ge b'$.  This means that there
exists a strategy profile $\vec{\tau}^k \in \Drm^k$ such that
$U_i(\vec{\tau}^k) - U_i(\sigma_i',\vec{\tau}^k_{-i}) \ge b'$.  
By the Bolzano-Weierstrass theorem%
, the sequence of
strategies  $\<\vec{\tau}^k\>_{k = 1, 2, \ldots}$ has a convergent
subsequence  $\<\vec{\tau}^{k_j}\>_{j = 1, 2, \ldots}$ that  
converges to some strategy profile $\vec{\tau}$.  Since $\tau^k \in
\Drm^m(\Scal)$ for all $k \ge m$, it must be the case that, except
for possibly a finite initial segment, this convergent subsequence is in
$\Drm^m(\Scal)$.  Since $\Drm^m(\Scal)$ is closed,
$\vec{\tau}$, the limit of the convergent subsequence, is in
$\Drm^m(\Scal)$ for all $m \ge 1$. Thus,
$\vec{\tau} \in \Drm^\infty(\Scal)$.  Now a simple continuity
argument shows that $U_i(\vec{\tau}) -
U_i(\sigma_i',\vec{\tau}_{-i}) \ge b' > b$, a contradiction.
\qed
\end{proof}
\olem{pd}
$\regret_1^{\Strat}(\sigmaad) = (n-1)(u_3 - u_2) + \max(-u_1,u_2-u_3)$.
Moreover, if $\strat$ is a strategy for player 1 where he
plays $c$ before  seeing player 2 play $c$ (i.e., where player 1 either
starts out playing $c$ or plays $c$ at the $k$th for $k>1$ move after seeing
player 2 play $d$ for the first $k-1$ moves), then  
$\regret_1^{\Strat}(\strat) > (n-1)(u_3 - u_2) + \max(-u_1,u_2-u_3)$.
\eolem

\begin{proof}
Let $\strat_c$ be the strategy where player 2 starts out playing $d$ and
then plays $c$ to the end of the game if player 1 plays $c$, and plays
$d$ to the end of the game if player 1 plays $d$.  We have
$\regret_1^{\Strat_1}(\sigmaad \mid \strat_c)  = (n-1)(u_3-u_1) - u_1$:
player 1 
gets gets $nu_1$ with $(\sigmaad,\strat_c)$, and could have gotten
$(n-1)u_3$ if he had cooperated on the first move and then always defected. 

Let $\strat_c'$ be the strategy where player 2
starts out playing $c$ and then plays $c$ to the end of the game if
player 1 plays $c$, and plays $d$ to the end of the game if player 1
plays $d$.  It is easy to see that 
$\regret_1^{\Strat_1}(\sigmaad \mid \strat_c') = 
(n-1)(u_3-u_1)  + (u_2 - u_3)$.  
Thus, $\regret_1^{\Strat}(\sigmaad) \ge (n-1)(u_3-u_2) +
\max(-u_1,u_2-u_3)$.  We now show that $\regret_1^{\Strat}(\sigmaad) = 
(n-1)(u_3-u_2) + \max(-u_1,u_2-u_3)$.  
For suppose that the regret is maximized if player 2
plays some strategy $\strat$, and player 1's best response to $\strat$ is
$\strat'$.  Consider the first place where the play of $(\strat',\strat)$
differs from that of $(\sigmaad,\strat)$. This must happen \emph{after}
a move where player 1 plays $c$ with $\strat'$.   For as long as player
1 plays $d$ with $\strat'$, player 2 cannot distinguish $\strat'$ from
$\sigmaad$, and so does the same thing in response.  So suppose that player 1
plays $c$ at move $k$ with $\strat'$.  If player 2 plays $c$ at step $k$,
player 1 gets a payoff of $u_3$ with $(\sigmaad,\strat)$ at step $k$ and a
payoff of $u_2$ with $(\strat',\strat)$.  Thus, player 1's total payoff with
$(\sigmaad,\strat)$ is at least $(n-1)u_1 + u_3$, while his payoff with
$(\sigmaad,\strat)$ is at most $(n-1)u_3 + u_2$; thus, his regret is at
most $(n-1)(u_3-u_1) + (u_2 - u_3)$. On the other hand, if player 2 plays $d$
with $\strat$ at step $k$, then player 1's payoff at step $k$ with
$(\strat',\strat)$ is 0, while his payoff at step $k$ with
$(\sigmaad,\strat)$ is $u_1$.  Thus, his regret is at most $(n-1)(u_3-u_1) 
-u_1$.  (In both cases, the regret can be that high only if $k=1$.)

We next show that if $\strat$ is a strategy for player 1 where he
plays $c$ before  seeing player 2 play $c$, then 
$\regret_1^{\Strat}(\strat) > (n-1)(u_3-u_2) + \max(-u_1,u_2-u_3)$. 
Suppose that $k$ is the first move where player 1 plays
$c$ despite not having seen $c$ before.  If $k=1$ (so that player 1
cooperates on the first move), let $\strat_d$ be the strategy
where player 2 plays $d$ for the first move, then plays $c$ to the
end of the game if player 1 has played $d$ for the first move, and
otherwise plays $d$ to the end of the game.  It is easy to see that 
Then $\regret_1^{\Strat_1}(\strat \mid\strat_d) =
(n-1)(u_3-u_1) + u_1$.  On the other hand, if $k > 1$, then the
regret $\regret_1^{\Strat_1}(\strat \mid\strat_c) 
 \ge (n-1)(u_3 - u_1)$.  Thus, $\regret_1^{\Strat}(\strat) > 
(n-1)(u_3 - u_2) + \max(-u_1,u_2-u_3)$.
\qed \end{proof}

\fullv{
\opro{onlypure.prop}
Let $G=(\Play,A,\vec{u})$ be a strategic game and 
let $\sigma_i$ be a mixed strategy for player $i$.
Then $\regret_{i}^{\Smixed}(\sigma_i) = \max_{\vec{a}_{-i} \in A_{-i}}
\regret^{A_{i}}(\sigma_i \mid \vec{a}_{-i})$.
\eopro

\begin{proof}
Note that, for all strategies profiles $\vec{\sigma}_{-i}$, there exists
some strategy $\sigma_i$ such that 
$U_i^{\Smixed}(\vec{\sigma}_{-i}) = U_i(\vec{\sigma})$.
It follows that 
$$\begin{array}{lll}
U_i^{\Smixed}(\vec{\sigma}_{-i}) &= &U_i(\vec{\sigma})\\
&= &\sum_{\vec{a}_{-i} \in A_{-i}} \vec{\sigma}_{-i} (\vec{a}_{-i})
U_i(\sigma_i, \vec{a}_{-i})\\
&\leq &\sum_{\vec{a}_{-i} \in A_{-i}} \vec{\sigma}_{-i} (\vec{a}_{-i})
U_i^{\Smixed}(\vec{a}_{-i}).
\end{array}$$
Thus
$$\begin{array}{ll}
&\regret^{\Smixed}(\sigma_i \mid \vec{\sigma}_{-i})\\
= &U_i^{\Smixed}(\vec{\sigma}_{-i}) - U_i(\sigma_i,\vec{\sigma}_{-i})\\  
= &U_i^{\Smixed}(\vec{\sigma}_{-i}) - \sum_{\vec{a}_{-i} \in A_{-i}}
\vec{\sigma}_{-i} (\vec{a}_{-1}) U_i(\sigma_i,\vec{a}_{-i})\\
\leq &\sum_{\vec{a}_{-i} \in A_{-i}} \vec{\sigma}_{-i} (\vec{a}_{-i})
U_i^{\Smixed}(\vec{a}_{-i}) - \sum_{\vec{a}_{-i} \in A_{-i}}
\vec{\sigma}_{-i} (\vec{a}_{-1}) U_i(\sigma_i,\vec{a}_{-i})\\  
= &\sum_{\vec{a}_{-i} \in A_{-i}} \vec{\sigma}_{-i} (\vec{a}_{-i})
\regret^{\Smixed}(\sigma_i \mid \vec{a}_{-i})\\
\leq &\max_{\vec{a}_{-i} \in A_{-i}} 
\regret^{\Smixed}(\sigma_i \mid \vec{a}_{-i}). 
\end{array}
$$
It follows that 
$$
\regret_{i}^{\Smixed}(\sigma_i) = \max_{\vec{\sigma}_{-i} \in
\Smixed_{-i}} \regret^{\Smixed}(\sigma_i \mid \vec{\sigma}_{-i}) =
\max_{\vec{a}_{-i} \in A_{-i}} \regret^{\Smixed}(\sigma_i \mid
\vec{a}_{-i}). \ \ \mbox{\qed}
$$
\end{proof}
}

\olem{lem:mixedregret}  $\regret_1^{\Sigma}(\sigma) < 3$.
\eolem

\begin{proof}
By Proposition~\ref{onlypure.prop}, to compute 
$\regret_1^{\Sigma}(\sigma)$, 
it suffices to compute $\regret_1^{\Sigma_2}(\sigma \mid a)$ for each
action $a$ of player 2.
If Player 1 plays $\sigma$ and player 2 player plays 100,
then the best response for player 1 is 99, giving him a payoff of 99.
The payoff with $\sigma$ is
$$\begin{array}{ll}
&100  \times 1/2+ 101 \times 1/4 + 100 \times 1/8 + \cdots + 5 \times
2^{-98} + 4 \times 2^{-98}\\
&102  \times 1/2+ 101 \times 1/4 + 100 \times 1/8 + \cdots + 5 \times
2^{-98} + 4 \times 2^{-98} - 1\\
= &4 \times (1/2 + 1/4 + \cdots + 2^{-98} + 2^{-98}) + \\
  &1 \times (1/2 + 1/4 + \cdots + 2^{-98}) + \\
  &1 \times (1/2 + 1/4 + \cdots + 2^{-97}) + \cdots  1/2 - 1\\
= &4 + (1-2^{-98} + (1-2^{-97}) + \cdots + (1 - 1/2)  - 1\\
= &102 - (1/2 + 1/4 + \cdots + 2^{-98}) - 1\\
= &100 + 1/2^{98},
\end{array}
$$
so the regret is less than 1.  
Similarly, if player 2 plays $k$ with $2 \le k \le 99$, 
the best response for player 1 is $k-1$, which would give player 1 a payoff of
$k+1$, while the payoff from $\sigma$ is 
$$\begin{array}{ll}
(k-2)(1/2 + \cdots + 1/2^{100 -k}) + k \times 1/2^{101-k} + (k+1) \times
1/2^{102 -k } + k \times 1/2^{103-k} +\\
 (k-1) \times 1/2^{104-k} + \cdots
+ 5 \times 1/2^{98} + 4 \times 1/2^{98}.
\end{array}$$
Thus, player 1's regret if player 2 plays $k$ is 
$$\begin{array}{ll}
&3 \times (1/2 + \cdots + 1/2^{100 -k} + 1/2^{105-k}) + 
2 \times 1/2^{104-k} + 
1 \times (1/2^{101-k} + 1/2^{103-k})\\ +
&1/2^{106-k}(4 + 5\times 1/2 + 6 \times 1/4 + \cdots + (k-4) \times 1
/2^{98} )\\
= & 3 \times (1/2 + \cdots + 1/2^{100 -k} + 1/2^{105-k})
2 \times 1/2^{104-k} + 
1 \times (1/2^{101-k} + 1/2^{103-k}) + 6 \times 1/2^{106-k} )\\
= & 3 \times (1/2 + \cdots + 1/2^{100 -k} + 1/2^{104 - k})
2 \times  1/2^{104-k} + 
1 \times (1/2^{101-k} + 1/2^{103-k})\\
\sim &3 \times (1 - 1/2^{101-k})
< 3.
\end{array}
$$
\qed \end{proof}

\bibliographystyle{chicagor}
\bibliography{z,joe} 

\end{document}